\documentclass[journal]{IEEEtran}

\usepackage{amsthm,amsmath,amssymb}
\usepackage{eucal}
\usepackage{xifthen}
\usepackage{mathtools}
\usepackage{enumerate}
\usepackage{microtype}
\usepackage{xspace}
\usepackage{bm}
\usepackage{cite}
\usepackage{fancyhdr}
\usepackage{lastpage}
\usepackage[small]{caption}
\usepackage{xcolor}
\allowdisplaybreaks

\long\def\comment#1{}

\newfont{\bbb}{msbm10 scaled 700}

\newfont{\bb}{msbm10 scaled 1100}
\newcommand{\CC}{\mbox{\bb C}}

\newcommand{\RR}{\mbox{\bb R}}

\newcommand{\mbs}[1]{\bm{#1}}
\newcommand{\vect}[1]{{\lowercase{\mbs{#1}}}}
\newcommand{\mat}[1]{{\uppercase{\mbs{#1}}}}

\renewcommand{\Bmatrix}[1]{\begin{bmatrix}#1\end{bmatrix}}

\newcommand{\Pmatrix}[1]{\begin{array}{ll}#1\end{array}}

\renewcommand{\Re}[1][]{\ifthenelse{\isempty{#1}}{\operatorname{Re}}{\operatorname{Re}\left(#1\right)}}
\renewcommand{\Im}[1][]{\ifthenelse{\isempty{#1}}{\operatorname{Im}}{\operatorname{Im}\left(#1\right)}}

\newcommand{\bv}{\vect{b}}

\newcommand{\ev}{\vect{e}}

\newcommand{\rv}{\vect{r}}
\newcommand{\sv}{\vect{s}}

\newcommand{\xv}{\vect{x}}
\newcommand{\yv}{\vect{y}}

\newcommand{\Am}{\mat{a}}
\newcommand{\Bm}{\mat{b}}
\newcommand{\Cm}{\mat{c}}

\newcommand{\Fm}{\mat{f}}

\newcommand{\Jm}{\mat{j}}

\newcommand{\Mm}{\mat{M}}

\newcommand{\Um}{\mat{u}}
\newcommand{\Vm}{\mat{V}}

\newcommand{\Xm}{\mat{x}}

\newcommand{\Ac}{{\mathcal A}}

\newcommand{\Cc}{{\mathcal C}}
\newcommand{\Dc}{{\mathcal D}}
\newcommand{\Ec}{{\mathcal E}}

\newcommand{\Gc}{{\mathcal G}}
\newcommand{\Hc}{{\mathcal H}}
\newcommand{\Ic}{{\mathcal I}}
\newcommand{\Jc}{{\mathcal J}}
\newcommand{\Kc}{{\mathcal K}}

\newcommand{\Nc}{{\mathcal N}}

\newcommand{\Qc}{{\mathcal Q}}
\newcommand{\Rc}{{\mathcal R}}
\newcommand{\Sc}{{\mathcal S}}
\newcommand{\Tc}{{\mathcal T}}
\newcommand{\Uc}{{\mathcal U}}
\newcommand{\Wc}{{\mathcal W}}
\newcommand{\Vc}{{\mathcal V}}
\newcommand{\Xc}{{\mathcal X}}

\newcommand{\Id}{\mat{\mathrm{I}}}
\newcommand{\one}{\mat{\mathrm{1}}}
\newcommand{\zero}{\mat{\mathrm{0}}}

\newcommand{\CN}[1][]{\ifthenelse{\isempty{#1}}{\mathcal{N}_{\mathbb{C}}}{\mathcal{N}_{\mathbb{C}}\left(#1\right)}}
\renewcommand{\P}[1][]{\ifthenelse{\isempty{#1}}{\mathbb{P}}{\mathbb{P}\left(#1\right)}}
\newcommand{\E}[1][]{\ifthenelse{\isempty{#1}}{\mathbb{E}}{\mathbb{E}\left(#1\right)}}
\renewcommand{\det}[1][]{\ifthenelse{\isempty{#1}}{\mathrm{det}}{\mathrm{det}\left(#1\right)}}
\newcommand{\trace}[1][]{\ifthenelse{\isempty{#1}}{\mathrm{tr}}{\mathrm{tr}\left(#1\right)}}
\newcommand{\rank}[1][]{\ifthenelse{\isempty{#1}}{\mathrm{rank}}{\mathrm{rank}\left(#1\right)}}
\newcommand{\diag}[1][]{\ifthenelse{\isempty{#1}}{\mathrm{diag}}{\mathrm{diag}\left(#1\right)}}

\DeclarePairedDelimiter\abs{\lvert}{\rvert}
\DeclarePairedDelimiter\Abs{\lvert}{\rvert^2}
\DeclarePairedDelimiter\norm{\lVert}{\rVert}

\renewcommand{\Re}{{\rm Re}}
\renewcommand{\Im}{{\rm Im}}

\newcommand{\defeq}{\triangleq}

\newtheorem{remark}{Remark}
\newtheorem{definition}{Definition}
\newtheorem{theorem}{Theorem}
\newtheorem{example}{Example}
\newtheorem{corollary}{Corollary}
\newtheorem{lemma}{Lemma}

\usepackage{hyperref}
\hypersetup{
    bookmarks=true,         
    unicode=false,          
    pdftoolbar=true,        
    pdfmenubar=true,        
    pdffitwindow=false,     
    pdfstartview={FitH},    
    pdfnewwindow=true,      
    colorlinks=true,       
    linkcolor=red,          
    citecolor=green,        
    filecolor=magenta,      
    urlcolor=cyan           
}

\usepackage{multirow}
\usepackage{mathrsfs}

\newcommand{\st}{{\rm s.t.}}
\newcommand{\sym}{{\rm sym}}
\newcommand{\sic}{{\rm SIC}}
\DeclareMathAlphabet{\mathcal}{OMS}{cmsy}{m}{n}

\begin{document}
\title{Topological Interference Management with \\Decoded Message Passing}
\author{\IEEEauthorblockN{Xinping Yi, {\em Member, IEEE} and Giuseppe Caire, {\em Fellow, IEEE}}\\ 
\thanks{This work has been presented in part at Proc. IEEE Int. Symp. Information Theory (ISIT'16), Barcelona, Spain, Jul. 2016.}
\thanks{X. Yi and G. Caire are with Communication and Information Theory Chair in Department of Electrical Engineering and Computer Science at Technische Universit{\"a}t Berlin, 10587 Berlin, Germany. (email: {\tt \{xinping.yi, caire\}@tu-berlin.de})}
}

\maketitle

\begin{abstract}
The topological interference management (TIM) problem studies {\em partially-connected} interference networks with no channel state information except for the network topology (i.e., connectivity graph) at the transmitters. In this paper, we consider a similar problem in the {\em uplink} cellular networks, while message passing is enabled at the receivers (e.g., base stations), so that the decoded messages can be routed to other receivers via backhaul links to help further improve network performance. 
For this TIM problem with decoded message passing (TIM-MP), we model the interference pattern by conflict digraphs, connect orthogonal access to the acyclic set coloring on conflict digraphs, and show that one-to-one interference alignment boils down to orthogonal access because of message passing. With the aid of polyhedral combinatorics, we identify the structural properties of certain classes of network topologies where orthogonal access achieves the optimal degrees-of-freedom (DoF) region in the information-theoretic sense. The relation to the conventional index coding with {\em simultaneous} decoding is also investigated by formulating a generalized index coding problem with {\em successive} decoding as a result of decoded message passing. The properties of reducibility and criticality are also studied, by which we are able to prove the linear optimality of orthogonal access in terms of symmetric DoF for the networks up to four users with all possible network topologies (218 instances). Practical issues of the tradeoff between the overhead of message passing and the achievable symmetric DoF are also discussed, in the hope of facilitating efficient backhaul utilization.
\end{abstract}

\section{Introduction}
\label{sec:Intro}
As the cellular network becomes larger, denser and more heterogeneous, interference management is increasingly crucial and challenging. The substantial gain promised by sophisticated interference management techniques (e.g., interference alignment \cite{Jafar:IA}) requires usually that (almost) perfect and instantaneous channel state information at the transmitters (CSIT) is accessible. Nevertheless, to obtain CSIT perfectly and instantaneously is challenging, if not impossible. Especially when the number of users/antennas is large or the channel changes rapidly, it will be expensive to obtain CSIT timely with reasonable accuracy. Relaxations of the perfect and instantaneous CSIT requirements have been investigated in various networks (e.g., instantaneous CSIT with limited accuracy \cite{Jindal:2006}, perfect but delayed CSIT \cite{MAT}). However, if only finite-precision CSIT is available, the system degrees of freedom (DoF) value, i.e., roughly speaking the number of non-interfering Gaussian channels that the system is able to support simultaneously, collapses to the situation as if no CSIT was available at all \cite{lapidoth:2006,imageset}. Indeed, with no CSIT, the transmitters cannot distinguish different receivers, and are totally blind. 

The no-CSIT assumption is somewhat too pessimistic. In fact, certain coarse channel information (e.g., channel fading statistics, strength, and users' locations) is easily obtained even in today's practical systems. For instance, if the fading channels of different users follow some structured patterns, then blind interference alignment could improve DoF beyond the absolutely no CSIT case \cite{Jafar:BIA}. In addition, the DoF collapse was observed under the assumption that the wireless network is fully connected, so that interference is everywhere no matter whether it is strong or weak enough to be negligible. Intuitively, it makes no sense for a system designer to take into account the interference from very far away base stations. As the interference power rapidly decays with distance for distances beyond some critical threshold due to shadowing, blocking, and Earth curvature, interference from some sources is inevitably weaker than others, which suggests the use of a partially-connected bipartite graph to model, at least approximately, the network topology.

Interference networks with no channel state information (CSI) except for the knowledge of the connectivity graph at the transmitters have been considered under the name of the ``topological interference management (TIM)'' problem \cite{Jafar:2013TIM}. It has been shown that substantial gains in terms of DoF can be obtained with only this topological information for partially-connected interference networks. Surprisingly, one half DoF per user, which is optimal for an interference channel with perfect and global CSIT, can be attained for some partially-connected interference channels with only topological information. Its substantial reduction of CSIT requirement has attracted a lot of followup works aiming at various aspects, such as the consideration of fast fading channels \cite{Avestimehr:2013TIM,Avestimehr:loserank}, alternating connectivity \cite{Sun:2013TIM,Sezgin:TIM}, multiple antennas \cite{MIMOTIM}, and cellular networks \cite{Jafar:1D,Gao:TIM,Yi:Fractional}.
The TIM problem was also nicely bridged to the ``index coding'' problem \cite{Index2011}, where the former offers well-developed interference management techniques (such as interference alignment) to attack the latter, and also serves as an intriguing application of great practical interest in wireless networks for the latter. 

Recently, the TIM problem under a broadcast setting with distributed transmitter cooperation in the {\em downlink} cellular network was considered in \cite{Yi:TIM}. It has been shown that, if message sharing is enabled at the base stations, higher rate transmission can be created by allocating messages to the transmitters in a way such that the interference can be perfectly avoided or aligned.
As a dual problem, a natural question then to ask is, whether receiver cooperation (or cooperative decoding at the base stations) in the {\em uplink} cellular networks also offers us some gains under the TIM setting.

Cooperative decoding at the base stations in the uplink cellular networks was widely studied (e.g., \cite{simeone2009local,ClusteredDec}), where the received signals are shared among base stations via backhaul links so that joint signal processing is enabled. Nevertheless, joint signal processing and decoding results in huge amount of backhaul overhead, even if the quantized received signal samples are shared locally in a clustered decoding fashion \cite{ClusteredDec}.
Most recently, a new type of local base station cooperation framework in uplink cellular networks was studied in \cite{Caire:CellularIA} to boost the overall network performance. Differently from the strategy of sharing quantized received signals, the authors in \cite{Caire:CellularIA} considered a successive decoding policy, in which the message at each receiver is decoded based on the locally received signal as well as the decoded messages passed from neighboring base stations that have already decoded their messages at an earlier stage. It has been shown that the local and single-round (non-iterative) message passing enables interference alignment without requiring symbol extensions or lattice alignment. 
For these results, it is crucial to exploit the partial connectivity of the interference graph while, as usual, the local interference alignment scheme requires perfect instantaneous CSIT.
A natural question is whether the CSIT requirement can be relaxed in the decoded message passing setting. More specifically, with decoded message passing, is it possible to attain performance gain in partially-connected cellular networks with only topological information?

In this work, we formally formulate the TIM problem with decoded message passing at the receivers, referred to as the ``TIM-MP'' problem. As soon as a receiver decodes its own message, it can pass its message to any other receivers who are interested. 
Building on this decoded message passing setting, we model the interference pattern by conflict digraphs, and connect orthogonal access to the acyclic set coloring on conflict digraphs. With the aid of polyhedral combinatorics, we identify certain classes of network topologies for which orthogonal access achieves the optimal DoF region. The relation to index coding is also studied by formulating a generalized index coding problem with successive decoding. Reducibility and criticality are also discussed in the hope of reducing large-size problems to smaller ones. By reducibility and criticality, the linear optimality of orthogonal access in terms of symmetric DoF is also shown for the small-size networks up to four users with all possible 218 non-isomorphic topologies. Practical issues for TIM-MP problems such as the tradeoff between the overhead of message passing and the achievable symmetric DoF are also discussed in the hope of facilitating the most efficient backhaul utilization.

More specifically, our contributions are organized as follows.
\begin{enumerate}
\item In Section \ref{sec:orth}, we model the interference pattern by a conflict directed graph, turning the interference between different transmitter-receiver pairs (i.e., the respectively desired messages) to the directed connectivity between nodes (representing the corresponding messages) in a directed graph. By this graphic modeling, we connect orthogonal access of TIM-MP problems to acyclic set coloring on conflict digraphs, where the latter is well-studied in the graph theory literature. The achievable symmetric DoF due to single-round and multiple-round message passing are connected to two graph theoretic parameters, dichromatic number and fractional dichromatic number. Thanks to the equivalence between local coloring and one-to-one alignment, we also show that one-to-one alignment boils down to orthogonal access as a result of decoded message passing, by proving that local acyclic set coloring is not better than acyclic set coloring.
\item We establish in Section \ref{sec:optimal-OA} the outer bound of the achievable DoF region by cycle and clique inequalities, and further connect it to set packing and covering polytopes. With the aid of polyhedral combinatorics, we identify sufficient conditions for which orthogonal access (i.e., fractional acyclic set coloring) achieves the optimal DoF region. Such conditions ensure the integrality of the outer bound of DoF region polytopes, where the integral extreme points of the polytopes can be achieved by acyclic set coloring. Time sharing among the integral extreme points yields the whole DoF region.  
\item The relation to index coding is also studied in Section \ref{sec:sic}, showing that TIM-MP corresponds to a generalized index coding problem, referred to as successive index coding (SIC). Generalizing conventional index coding, SIC allows successive decoding at the receivers, where as soon as a receiver decodes its desired message, it can declare it and pass it to other receivers as additional side information. The decoding and message passing orders play an crucial role, which makes SIC a more complex combinatorial problem. The analogous coding schemes to (partial) clique covering are also given. The vertex-reducibility and arc-criticality of SIC are also investigated, by which SIC problems with large vertex/arc size can be reduced to ones with smaller size.
\item The linear optimality of orthogonal access is considered in \ref{sec:linear-opt}. We first consider some special network topologies that do not satisfy the sufficient conditions in Section \ref{sec:optimal-OA}, and then prove the linear optimality of orthogonal access with respect to symmetric DoF, if restricted to linear schemes. Thanks to the vertex-reducibility and arc-criticality investigated in Section \ref{sec:sic}, the linear optimality of symmetric DoF or broadcast rate for small networks up to 4 users with all possible network topologies are fully characterized by orthogonal access.
\item The practical issue on the tradeoff between achievable symmetric DoF and the number of passed messages is also discussed in Section \ref{sec:tradeoff}. We identify a sufficient condition under which only one message passing is helpful to improve the DoF region. The tradeoff between achievable symmetric DoF and the overhead of message passing is formulated as a matrix completion problem, which can be solved algorithmically although closed-form solution remains a challenge.
\end{enumerate}

\underline{\bf Notations}:
Throughout this paper, we define $\Kc \defeq \{1,2,\dots,K\}$, and $[n] \defeq \{1,2,\dots,n\}$ for any integer $n$. Let $A$, $\Ac$, and $\Am$ represent a variable, a set, and a matrix, respectively. In addition, $\Ac^c$ is the complementary set of $\Ac$, and $\abs{\Ac}$ is the cardinality of the set $\Ac$. The set  $x(\Sc)$ or $x_{\Sc}$ represents a set or tuple $\{x_i, i \in \Sc\}$ indexed by $\Sc$.  $\Am_{ij}$ represents the $ij$-th entry of the matrix $\Am$. Define $\Ac \backslash a \defeq \{x| x \in \Ac, x \neq a\}$ and $\Ac_1 \backslash \Ac_2 \defeq \{x | x \in \Ac_1, x \notin \Ac_2\}$. $\zero$ and $\one$ represent respectively the all-zero and all-one vectors.

\section{System Model}
\label{sec:sys}
\subsection{Channel Model}
We consider the uplink of a cellular network with $K$ user terminals (i.e., transmitters) that want to send messages to $K$ base stations (i.e., receivers), respectively. The base stations are connected with backhaul links, through which one base station could pass its own decoded message to its neighboring ones.
Both user terminals and base stations are equipped with a single antenna each. It is assumed that, due to the scarce channel state feedback resource, the users have no access to channel realizations but only know the network connectivity graph, i.e., which user is connected to which base station. 
The received signals in this partially-connected network are modeled, for the base station $j$ at time instant $t$, by
\begin{align}
Y_j(t) = \sum_{i \in \Tc_j} h_{ji}(t) X_i(t) + Z_j(t)
\end{align}
where $X_i(t)$ is the transmitted signal subject to the average power constraint $\E(\Abs{X_i}) \le P$, $Z_j(t)$ is the Gaussian noise with zero-mean and unit-variance at the base stations, and $h_{ji}(t)$ is the channel coefficient between user $i$ and base station $j$, and is not known by the users.
Here $\Tc_j$ represents the transmit set containing the indices of users that are {\em connected} to base station $j$, for $ j \in \Kc \defeq \{1,2,\dots,K\}$. We point out that channel coefficients $\{h_{ji}(t), \forall~i,j,t\}$ are not available at the users, yet the network topology (i.e., $\Tc_j, \forall j$) is known by both users and base stations.
The network topology is assumed to be fixed throughout the duration of communication. Such a setup is referred to as the ``Topological Interference Management (TIM)'' setting.

\subsection{Problem Statement}
Similarly to the definition in \cite{Caire:CellularIA}, a decoding order $\pi$ is a partial order $\prec_{\pi}$ such that $i \prec_{\pi} j$ indicates the message $W_i$ should be decoded before $W_j$.
We assume that the base station $i$ only decodes its own desired message $W_i$, and then passes it to the base station $j$, even though sometimes the messages desired by other base stations are also decodable.
As anticipated in Section \ref{sec:Intro}, we refer to this setting combining TIM and decoded message passing as ``TIM-MP''.
In the TIM-MP problem, given a decoding order $i \prec_{\pi} j$, once $W_i$ is decoded, it can be passed to receiver $j$ to help decoding $W_j$. Throughout this paper, we consider unconstrained message passing, that is, a message can be passed to any other receivers who are interested \footnote{In fact, because of the locality of interference in physically motivated network topologies, only the neighboring receivers suffer from interference from message $W_i$, and therefore messages are passed in a neighbor-propagation fashion.}.

Formally, the achievable rate of the TIM-MP problem can be defined as follows.
For a network with topology represented by the bipartite graph $\Gc$, a rate tuple $(R_1^{\pi},\dots,R_K^{\pi})$ is said to be achievable under a specified decoding order $\pi$ if there exists a $(2^{nR_1^{\pi}},\dots,2^{nR_K^{\pi}},n)$ coding scheme consists of the following elements:
\begin{itemize}
\item 
$K$ message sets $\Wc_i \defeq [1:2^{nR_i^{\pi}}]$, from which the message $W_i$ is uniformly chosen, $\forall~i \in \Kc$;
\item
$K$ encoding functions $f_{i}^{(n,\pi)}: \Wc_i \times \Gc \mapsto \CC^{n}$, $\forall~i \in \Kc$:
\begin{align}
X_i^n = f_{i}^{(n,\pi)} (W_i, \Gc)
\end{align}
with power constraint $\E [\Abs{X_i}] \le P$, where each transmitter has only access to its own message and the network topology graph $\Gc$;
\item
$K$ decoding functions $g_j^{(n,\pi)}: \CC^n \times \Sc_j^{\pi} \times \Gc \times \CC^{K \times K} \mapsto \Wc_j$, $\forall~j \in \Kc$:
\begin{align}
\hat{W}_j = g_j^{(n,\pi)} (Y_j^n, S_j^{\pi}, \Gc, \Hc)
\end{align}
where $\Hc=\{h_{ji}(t), \forall i,j,t\}$, and $S_j^{\pi}$ is a set of decoded messages passed from other receivers through backhaul links, defined as
\begin{align}
S_j^{\pi} = \{\hat{W}_k: k \prec_{\pi} j \};
\end{align}
\end{itemize}
such that the decoding error
$
P_e^{(n,\pi)} = \max_j \{\P (W_j \neq \hat{W}_j)\}
$
tends to zero when the code block length $n$ tends to infinity. 

We consider the following message passing policies:
\begin{itemize}
\item Single-round message passing: For a decoding order $\pi$, for all $(i,j)$ such that $i \prec_{\pi} j$, it is allowed to pass the messages only from receiver $i$ to receiver $j$.
\item Multiple-round message passing: It consists of a sequence of multiple single-round message passing. Different rounds can have distinct decoding orders. For instance, it may happen that $i  \prec_{\pi} j$ in one round and $j  \prec_{\pi'} i$ in another round with $\pi \neq \pi'$.
\end{itemize}

The achievable rate region with respect to the decoding order $\pi$, denoted as $\Rc^{\pi}$, is the set of all achievable rate tuples $(R_1^{\pi},\dots,R_K^{\pi})$, corresponding to the single-round message passing.
The capacity region with multiple-round message passing over all possible decoding orders is given by
\begin{align}
\mathscr{C}={\rm conv} (\cup_{\pi} \Rc^{\pi}),
\end{align} 
which can be obtained by time sharing among multiple single-round message passing with different decoding orders allowed in different rounds.

We follow the TIM setting and use symmetric DoF and DoF region as our main figures of merit. 
\begin{definition} [Symmetric DoF and DoF Region]
\begin{align}
d_{\sym} &= \limsup_{P \to \infty}  \sup_{(R,\dots,R) \in \mathscr{C}} \frac{R}{\log P} \\
{\mathscr D}&=\left\{(d_{\Kc}) \in \RR_+: d_i = \limsup_{P \to \infty}  \frac{R_i}{\log P}, \forall i  \right. \nonumber \\ & \qquad \qquad \left. {\rm s.t.} \quad (R_1,\dots,R_K) \in \mathscr{C} \right\}
\end{align}
\end{definition}

\subsection{Interference Modeling}
We model the mutual interference in the network as a directed message conflict graph. A directed graph (digraph) $\Dc=(\Vc,\Ac)$ consists of a set of vertices $\Vc$ and a set of arcs $\Ac$ between two vertices. We denote by $(u,v)$ the arc (i.e., directed edge) from vertex $u$ to vertex $v$. More graph theoretic definitions are presented in Appendix \ref{appendix:graph}. 

\begin{definition}[Conflict Digraph]
For a network topology, its directed conflict graph (briefly referred to as``conflict digraph'') is a digraph $\Dc=(\Vc,\Ac)$ such that $i \in \Vc$ represent the message $W_i$ from transmitter $i$ to receiver $i$ and $(i,j) \in \Ac$ represents the interfering link from transmitter $i$ to receiver $j$ in the interference network. 
\end{definition}
The conflict digraphs indicate not only the message conflict due to mutual interference, but also the source and the sink of the interference.
The conflict digraph captures exactly every instance of network topology. We refer to the TIM-MP problem with a specific conflict digraph as a TIM-MP instance.

\section{Orthogonal Access}
\label{sec:orth}
Orthogonal access is the simplest transmission scheme of practical interest. For the TIM problem, orthogonal access is to schedule independent sets of the conflict graph across time or frequency \cite{Yi:Fractional}, because simultaneous transmission of the messages in an independent set and orthogonal transmission of different independent sets across time or frequency avoid mutual interference. By contrast, message passing offers the possibility of interference cancelation for the messages that are not in an independent set. This complicates orthogonal access under the TIM-MP setting, as both interference avoidance and cancelation should be taken into account.

In what follows, we first introduce the concept of orthogonal access in the TIM-MP problem, and propose an inner bound of symmetric DoF for the single-round message passing setting, followed by the extension to the multiple-round message passing. The uselessness of one-to-one interference alignment is also shown from a graph theoretic perspective.

\subsection{What is Orthogonal Access?}
Instead of scheduling independent sets in the TIM problem, we schedule acyclic set in the TIM-MP problem, where the acyclic set is an induced sub-digraph that contains no directed cycles (dicycles). More properties of the acyclic set can be found in Appendix \ref{appendix:graph}. Thus, we have the following definition.
\begin{definition}[Orthogonal Access]
Orthogonal access in the TIM-MP setting consists of scheduling orthogonally acyclic sets of conflict digraphs across time or frequency.
\end{definition}

The messages in an acyclic set can be decoded successively via decoded message passing in one time slot. \footnote{Note that the propagation of the messages over the backbone is much faster than the transmission over the wireless interface. Therefore, we may treat it, for conceptual simplicity, as one time slot counting the time to transmit (simultaneously) the codewords, and neglecting the message passing propagation time. In practice, a sequence of codewords can be multiplexed in time for the same acyclic set and decoding order, such that the propagation can be done in a pipelined way, such that effectively the time needed for end to end propagation of the decoded messages is much less than the duration of transmission in the state defined by the acyclic set.} In an acyclic digraph $\Dc=(\Vc,\Ac)$, there always exists a topological ordering of $\Vc$ such that a vertex $u \in \Vc$ comes before a vertex $v \in \Vc$ if there is an arc $(u,v) \in \Ac$. Such a topological order gives us the decoding and message passing order. In particular, there exists at least one vertex with no incoming arcs in an acyclic set. We start with the decoding of these messages, which are free of interference. After decoding these messages, they are passed to the next ones which are only interfered by them such that the interference can be fully canceled out, and thus these next messages are also decodable. Keep doing this until all messages in this acyclic set are decoded in such a successive way. As such, simultaneous transmission of the messages in an acyclic set does not have residual interference left after the interference cancelation with passed messages.

Let us look at orthogonal access from a graph coloring perspective. If each acyclic set is assigned with one color, orthogonal access is equivalent to acyclic sets coloring of conflict digraphs. The messages assigned to the same color are simultaneously transmitted and successively decoded, whereas different colors are multiplexed over different time slots. 

\begin{definition}[Dichromatic Number \cite{Dichrom,CircularNo}]
The dichromatic number of a digraph $\Dc$, denoted by $\chi_A(\Dc)$, is the minimum number of colors required to color the vertices of $\Dc$ in such a way that every set of vertices with the same color induces an acyclic sub-digraph in $\Dc$.
\end{definition}

By this definition, we can immediately obtain an inner bound of symmetric DoF.
\begin{lemma} \label{theorem:OA}
For the TIM-MP instance with conflict digraph $\Dc$, we have an inner bound
\begin{align}
d_{\sym} \ge \frac{1}{\chi_A(\Dc)}
\end{align}
which is achieved by orthogonal access with single-round message passing.
\end{lemma}
\begin{remark}
\normalfont
Orthogonal access with single-round message passing can be seen as assigning a standard basis vector to each acyclic set. For instance, the assigning of the $i$-th column of an identity matrix to an acyclic set is equivalent to the scheduling of this acyclic set in $i$-th time slot.
\end{remark}

By Lemma \ref{theorem:OA}, we have the following two corollaries, and relegate the proofs to Appendix~\ref{proof:cor-1dof} and \ref{proof:cor-half-dof}.
\begin{corollary} \label{cor:1dof}
For TIM-MP instances, the optimal symmetric DoF is $d_{\sym}=1$ if and only if the conflict digraphs are acyclic.
\end{corollary}

\begin{corollary} \label{cor:half-dof}
For TIM-MP instances, if only single-round message passing is allowed, the optimal symmetric DoF is $d_{\sym}=\frac{1}{2}$ if the conflict digraph contains either only directed odd cycles or only directed even cycles.
\end{corollary}

\begin{remark}
\normalfont
The condition in Corollary \ref{cor:half-dof} is only sufficient but not necessary. There exists a larger family of conflict digraphs with $\chi_A=2$. There is a conjecture \cite{CircularNo} in the graph theory literature, claiming that the planar digraphs without length-2 dicycles have $\chi_A=2$. 
\end{remark}

\begin{example}
\normalfont
For the conflict digraphs in Fig. \ref{fig:half-dof}, there are only directed odd cycles in Fig. \ref{fig:half-dof}(a), only directed even cycles in Fig. \ref{fig:half-dof}(b), and both odd and even dicycles in Fig. \ref{fig:half-dof}(c) and (d). So, according to Corollary \ref{cor:half-dof}, the optimal symmetric DoF value with single-round message passing for both (a) and (b) is $\frac{1}{2}$. We cannot expect that DoF $\frac{1}{2}$ is achievable in (c). In fact $d_{\sym}=\frac{1}{3}$ is optimal, which will be shown later. Nevertheless, the optimal symmetric DoF value of Fig. \ref{fig:half-dof}(d) is also $\frac{1}{2}$, although the conflict digraph contains both even and odd cycles. This shows that the condition in Corollary \ref{cor:half-dof} is only sufficient but not necessary. \hfill $\lozenge$
\begin{figure}[htb]
 \centering
\includegraphics[width=1\columnwidth]{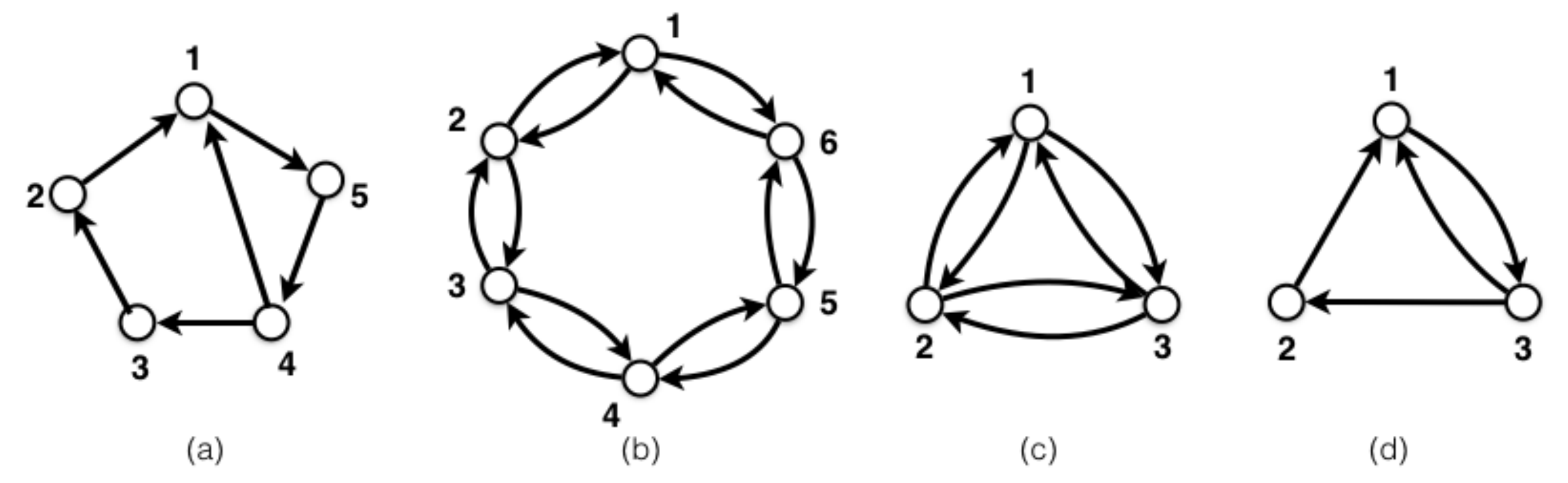}
\caption{The conflict digraphs that contain (a) only directed odd cycles, (b) only directed even cycles, and (c), (d) both directed odd and even cycles.}
\label{fig:half-dof}
\end{figure}
\end{example}

The dichromatic number of a digraph $\Dc$ can be represented as the solution to the following linear program:
\begin{subequations} \label{eq:lp-acyclic}
\begin{align}
\chi_{A}(\Dc) = \min & \sum_{A \in \mathscr{A}(\Dc)} g(A) \\
\st &\sum_{A \in \mathscr{A}(\Dc,v)} g(A) \ge 1, \quad \forall v \in \Vc(\Dc)\\
& g(A) \in \{0,1\}.
\end{align}
\end{subequations}
where $\mathscr{A}(\Dc)$ is the collection of all possible acyclic sets, and $\mathscr{A}(\Dc,v)$ is the collection of all possible acyclic sets that involve the vertex $v$. By relaxing $g(A) \in \{0,1\}$ to $g(A) \in [0,1]$, as the fractionalized versions of other graph theoretic parameters, the linear program \eqref{eq:lp-acyclic} yields the fractional dichromatic number $\chi_{A,f}(\Dc)$, which can also serve as an inner bound of the symmetric DoF.

\begin{lemma} \label{lemma:fractional}
For the TIM-MP instance with conflict digraph $\Dc$, we have
\begin{align}
d_{\sym} \ge \frac{1}{\chi_{A,f}(\Dc)},
\end{align}
which is achieved by orthogonal access with multiple-round message passing.
\end{lemma}

Fractional coloring, no matter whether independent or acyclic set coloring, can be treated as time sharing among a set of proper non-fractional coloring, where e.g., $g(A) \in [0,1]$ is the portion of the shared time of the acyclic set $A$. Thus, fractional acyclic set coloring of conflict digraphs is equivalent to orthogonal access with multiple-round message passing.
If there only exists the symmetric part in the conflict digraph (i.e., without uni-directed arcs, see Appendix \ref{appendix:graph}), both dichromatic number and its fractional version reduce to their counterparts in the underlying undirected graph, as acyclic sets in the digraph reduce to independent sets in the underlying undirected graph.

The $(K,L)$ regular network \cite{Yi:TIM} has a connectivity pattern where each receiver is connected to its paired transmitter and the next $L-1$ successive ones. By Lemma \ref{lemma:fractional}, we have the following corollary for the symmetric DoF inner bound for the regular network, whose proof is relegated to Appendix \ref{proof:cor-regular}.
\begin{corollary} \label{cor:regular}
For the $(K,L)$ regular network $(K \ge L)$ with $\Tc_{j}=\{j, j+1, \dots, j+L-1\} \mod K$, we have
\begin{align}
d_{\sym} \ge \frac{K-L+1}{K}
\end{align}
which is achieved by orthogonal access with multiple-round message passing.
\end{corollary}

Unless otherwise specified, orthogonal access in the rest of this paper is referred to multiple-round message passing. 

\subsection{Can Interference Alignment Help?}
Beyond the achievability through acyclic set coloring, one interesting question to ask is, if the more sophisticated achievability schemes, such as interference alignment, can outperform orthogonal access.

Roughly speaking, (subspace) interference alignment is to associate each interference with a subspace such that the superposition of interferences occupies a reduced dimensional subspace. One-to-one interference alignment is a special case of subspace alignment. It consists of aligning the interferences in a one-to-one manner, that is, given a one-dimensional subspace, two interferences are either completely aligned or disjoint.

{Under the TIM setting, orthogonal access is equivalent to fractional vertex coloring on the undirected conflict graph \cite{Jafar:2013TIM,Yi:Fractional}, and one-to-one interference alignment is a generalized version of orthogonal access \cite{Jafar:2013TIM}. Let us associate each  transmitter-receiver pair (i.e., each vertex in conflict graph) with a transmission scheduling vector of length $L$, where $L$ is  the number of scheduling intervals to properly serve all transmitter-receiver pairs without causing mutual interference (i.e., the number of colors for a proper vertex coloring on the conflict graph). Orthogonal access corresponds to assigning the basis vector $\ev_t$ (i.e., the $t$-th column of $\Id_L$) to the transmitter-receiver pairs associated to the color $t$, 
while one-to-one alignment consists of assigning general linearly independent vectors such that each receiver can recover its desired message by solving a linear system of equations (in the absence of noise). In general, this allows for a vector dimension $T  \le L$, such that interference alignment may improve over orthogonal access. At a given receiver, only the transmitters that cause interference do appear in the linear system. Accordingly, the required vector dimension $T$ depends merely on the number of different colors in the in-neighborhood of the directed conflict graph. Thus, a feasible one-to-one interference alignment scheme under the TIM setting is equivalent to a proper local coloring on the directed conflict graph \cite{localcoloring2013}.
}

This also analogously applies to the TIM-MP setting.
Given a proper acyclic set coloring for a conflict digraph, for an acyclic set, the maximum number of acyclic sets with different colors in the in-neighborhood (i.e., causing interference) does matter. In other words, the spanned subspace by the assigned vectors of these acyclic sets in the in-neighborhood should be minimized to make interference as aligned as possible.   

Analogously to \cite{localcoloring2013}, we introduce a local version of fractional acyclic set coloring.
\begin{definition}[Local Dichromatic Number] \label{def:loc-dichrom}
The local dichromatic number $\chi_{LA}(\Dc)$ of a digraph $\Dc$ is defined as
\begin{align}
\chi_{LA}(\Dc) \defeq \min_{c} \max_{v \in \Vc} \abs{\left\{c(u): u \in  \Nc^-_v\right\}}
\end{align}
where , $\Nc_v^-$ is the set of vertices in the closed in-neighborhood of $v$, and the minimum is over all possible acyclic set coloring $c: \Vc \mapsto \mathbb{N}$. The closed in-neighborhood $\Nc_v^-$ is defined as
\begin{align}
\Nc_v^- \defeq  \{v\} \cup \left\{ u: (u,v) \in {\Ac}(\Dc) \right\}.
\end{align}
\end{definition}
The local dichromatic number of a digraph $\Dc$ can be also represented as the solution to the following linear program:
\begin{subequations} \label{eq:local-coloring}
\begin{align}
\chi_{LA}(\Dc) = \min & \max_{v \in \Vc}  \sum_{A \in \mathscr{A}(\Dc):~ A \cap \Nc_v^- \neq \emptyset } g(A) \\
\st &\sum_{A \in \mathscr{A}(\Dc,v)} g(A) \ge 1, \quad \forall v \in \Vc(\Dc)\\
& g(A) \in \{0,1\}.
\end{align}
\end{subequations}
 
Its fractional version $\chi_{LA,f}(\Dc)$ can be similarly defined by replacing $g(A) \in \{0,1\}$ with $g(A) \in [0,1]$. 
It is clear that fractional local coloring is built upon the feasible fractional coloring of acyclic sets, and the difference is that the local coloring only counts colors in the closed in-neighborhood. So, we always have $\chi_{LA,f}(\Dc) \le \chi_{A,f}(\Dc)$, because an additional condition is imposed on the local version.

The linear program formulation of fractional local acyclic set coloring in \eqref{eq:local-coloring} is a straightforward extension of fractional local independent set coloring, where the acyclic sets in \eqref{eq:local-coloring} replace the independent sets. Similarly to the equivalence between interference alignment and local coloring shown in \cite{localcoloring2013},  it follows immediately that one-to-one interference alignment with message passing is equivalent to fractional local acyclic set coloring.
Thus, we have a new inner bound for the symmetric DoF due to interference alignment.
\begin{lemma}
\label{theorem:local}
For the TIM-MP instance with conflict digraph $\Dc$, we have
\begin{align}
d_{\sym} \ge \frac{1}{\chi_{LA,f}(\Dc)},
\end{align}
which is achieved by one-to-one interference alignment.
\end{lemma}

By Lemma \ref{lemma:IA-nohelp}, we show that one-to-one interference alignment does not help when decoded message passing is enabled, and present the proof in Appendix \ref{proof:IA-nohelp}.
\begin{lemma} \label{lemma:IA-nohelp}
With message passing, one-to-one interference alignment boils down to orthogonal access due to
\begin{align}
\chi_{LA,f}(\Dc)=\chi_{A,f}(\Dc).
\end{align}
\end{lemma}

\begin{example}
\normalfont
Let us consider an instance of the TIM-MP problem with $K=5$. The conflict digraph is shown in Fig. \ref{fig:frac_local}(a), where a proper fractional acyclic set coloring is given for acyclic sets $\{1,3\},\{2,4\}, \{3,5\}, \{4,1\}, \{5,2\}$. It requires in total 5 colors for these acyclic sets, each of which receives 2 colors, such that the sub-digraph induced by the vertices with the same color is acyclic. Fig. \ref{fig:frac_local} shows the in-neighborhood of the acyclic set $\{5,2\}$, which includes all other vertices. The fractional dichromatic number is $\frac{5}{2}$, which agrees with the fractional local dichromatic number. \hfill $\lozenge$
\begin{figure}[htb]
\centering
\includegraphics[width=0.9\columnwidth]{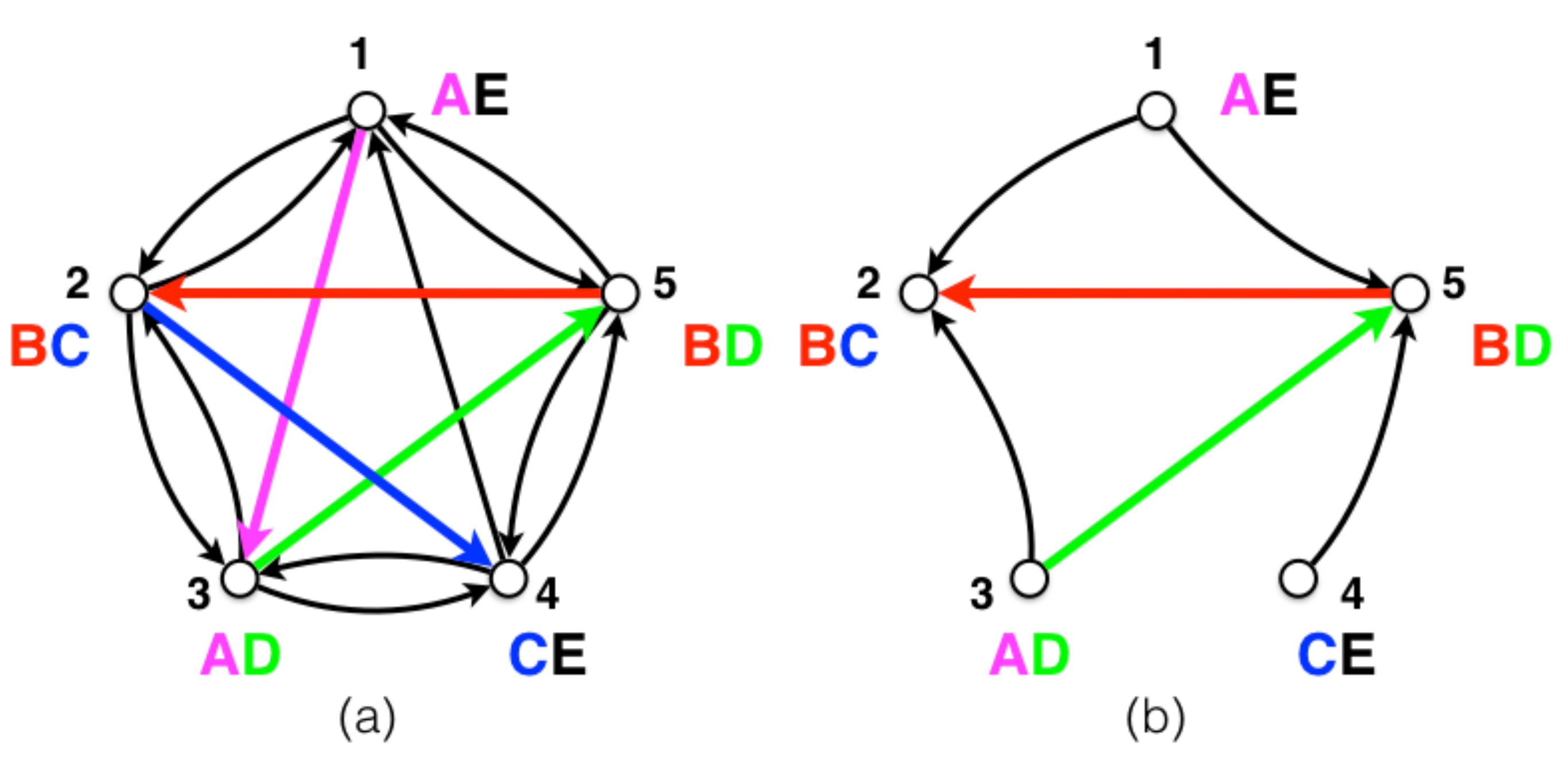}
\caption{ A conflict digraph contains both odd and even dicycles, and its symmetric part is a hole. (a) The fractional local acyclic set coloring, and (b) the in-neighborhood of an acyclic set $\{5,2\}$ that includes all vertices.}
\label{fig:frac_local}
\end{figure}
\end{example}

Although one-to-one interference alignment does not outperform orthogonal access, it remains as an interesting and challenging open problem to exploit the potential benefit of subspace interference alignment, which has been shown to provide further gains in problems such as multiple groupcast TIM \cite{Jafar:2013TIM}.

\section{The Optimality of Orthogonal Access via Polyhedral Combinatorics}
\label{sec:optimal-OA}
Having assessed that one-one-one interference alignment under message passing decoding boils down to orthogonal access, a natural question is how powerful orthogonal access is, and under what condition orthogonal access is DoF-optimal in the information-theoretic sense. 
Before proceeding further, we introduce some outer bounds.
The preliminaries related to polyhedral combinatorics can be found in Appendix \ref{appendix:polyhedral}.

\subsection{Outer Bounds via Polyhedral Combinatorics}
By the nature of message passing, we conclude that cliques and dicycles are main obstacles, and thus have the following outer bounds.
\begin{theorem}[Clique-Cycle Outer Bounds]
\label{theorem:outer-bounds}
The DoF region $\mathscr{D}$ of the TIM-MP problem is outer-bounded by
\begin{align} \label{eq:outer-bounds}
\mathscr{D} \subseteq \left\{(d_{\Kc}) : \Pmatrix{0 \le d_k \le 1, \forall k \in \Kc \\ \sum_{k \in Q} d_k \le 1, \; \forall Q \in \Qc \\ \sum_{k \in C} d_k \le \abs{C}-1, \; \forall C \in \Cc} \right\}
\end{align}
where $\Cc$ is the collection of all minimal dicycles (i.e., dicycles without chord), and $\Qc$ is the collection of all maximal cliques (i.e., cliques not a sub-digraph of other cliques).
\end{theorem}

\begin{proof}
See Appendix \ref{proof:outer-bounds}.
\end{proof}
\begin{remark}
\normalfont
We refer hereafter to the inequalities in \eqref{eq:outer-bounds} as individual inequalities, cycle inequalities, and clique inequalities, respectively. 
The cliques with size 1 are vertices, so clique inequalities imply the individual ones. The clique with size 2 is also a dicycle, such that clique and cycle inequalities have some inequalities in common. As all sub-digraphs of a clique are still cliques and the clique inequality of the maximal one implies all other ones, we only count the clique inequality with the maximal size. Moreover, if a dicycle has a chord, whatever its direction is, there exists a subset of vertices that form a shorter dicycle, rendering the constraint associated with the larger one redundant. As such, we only count the cycle inequalities corresponding to the dicycles without chord. \hfill $\lozenge$
\end{remark}

The outer bound with only individual and clique inequalities can be formed, by replacing $d_k$ by $x_k$, as a set packing polytope (see Appendix \ref{appendix:polyhedral})
\begin{align} \label{eq:packing}
\mathscr{P}(\Qc,x_{\Kc})=\left\{(x_{\Kc}): \Pmatrix{0 \le x_k \le 1, \forall k \in \Kc \\
\sum_{k \in Q} x_k \le 1, \; \forall Q \in \Qc } \right\}.
\end{align}
The cycle inequalities, with a replacement of variables $y_k=1-x_k$, can be equivalently rewritten as $\sum_{k \in C} y_k \ge 1$, and thus the outer bound with only individual and cycle inequalities can be formed, by replacing $d_k$ by $1-y_k$, as a set covering polytope (see Appendix \ref{appendix:polyhedral})
\begin{align} \label{eq:covering}
\mathscr{P}(\Cc,y_{\Kc}) = \left\{(y_{\Kc}): \Pmatrix{0 \le y_k \le 1, \forall k \in \Kc \\ \sum_{k \in C} y_k \ge 1, \; \forall C \in \Cc } \right\}.
\end{align}

Taking all individual, clique, and cycle inequalities into account, we have the outer bound formed, by replacing $d_k$ by $1-y_k$ and removing redundant inequalities, as the mixed set covering and packing polytope \cite{kiraly2015extension}
\begin{align} \label{eq:mixed}
\mathscr{P}(\Cc',\Qc',y_{\Kc})=\left\{(y_{\Kc}) : \Pmatrix{0 \le y_k \le 1, \forall k \in \Kc \\ \sum_{k \in C} y_k \ge 1, \; \forall C \in \Cc' \\ \sum_{k \in Q} y_k \ge \abs{Q}-1, \; \forall Q \in \Qc'} \right\}
\end{align}
where
\begin{gather}
\Cc' = \{ C: \abs{C} \ge 2, \forall C \in \Cc\} \\ \Qc'=\{Q: \abs{Q} \ge 3, \forall Q \in \Qc\} \\ \abs{C \cap Q} \le 1, \forall C \in \Cc', \forall Q \in \Qc'.
\end{gather}
The conditions $\abs{C} \ge 2$ and $\abs{Q} \ge 3$ are to ensure that the redundancy between clique and cycle inequalities is removed. For some $C \in \Cc$ and $Q \in \Qc$, if $\abs{C \cap Q} \ge 2$, then the condition $\sum_{k \in C} y_k \ge 1$ is redundant, so we add $\abs{C \cap Q} \le 1$ to avoid redundancy.

To rewrite the set packing and covering polytope into compact forms, we introduce two incidence matrices (see definitions in Appendix \ref{appendix:polyhedral}).

\begin{definition}[Clique-Vertex Incidence Matrix]
Let $\Qc$ be the collection of all induced maximal cliques of a digraph $\Dc=(\Vc,\Ec)$. The corresponding clique-vertex incidence matrix $\Am$ is a $\abs{\Qc} \times \abs{\Vc}$ binary matrix, where 
\begin{align}
\Am_{ij}=\left\{ \Pmatrix{1, & \text{if $v_j \in Q^i$}, \\ 0, & \text{otherwise} } \right.
\end{align}
where $v_j \in \Vc$, and $Q^i \in \Qc$ is the $i$-th clique in $\Qc$.
\end{definition}

\begin{definition}[Dicycle-Vertex Incidence Matrix]
Let $\Cc$ be the collection of all induced minimal dicycles of a digraph $\Dc=(\Vc,\Ec)$. The corresponding dicycle-vertex incidence matrix $\Bm$ is a $\abs{\Cc} \times \abs{\Vc}$ binary matrix, where 
\begin{align}
\Bm_{ij} = \left\{ \Pmatrix{1, & \text{if $v_j \in C^i$}, \\ 0, & \text{otherwise} } \right.
\end{align}
where $v_j \in \Vc$, and $C^i \in \Cc$ is the $i$-th dicycle in $\Cc$.
\end{definition}

As all clique inequalities correspond only to the maximal cliques, there are no dominating rows in the clique vertex incidence matrix $\Am$. As all cycle inequalities corresponds only to the dicycles without chord, there are no dominating rows in the cycle-vertex incidence matrix $\Bm$. Nevertheless, there might be dominating rows in the concatenation of $\Am$ and $\Bm$.

By the above two incidence matrices, the compact representation of set packing and covering polytopes can respectively represented as
\begin{align} \label{eq:packing}
\mathscr{P}(\Qc, \xv)&=\left\{\xv \in \mathbb{R}^K: \zero \le \xv \le \one, \Am \xv \le \one \right\} \\
\mathscr{P}(\Cc, \yv)&=\left\{\yv \in \mathbb{R}^K: \zero \le \yv \le \one, \Bm \yv \ge \one \right\}.
\end{align}

According to polyhedral combinatorics (see Appendix \ref{appendix:polyhedral}), the matrix $\Am$ is perfect if and only if $\mathscr{P}(\Qc, \xv)$ has only integral extreme points, and the matrix $\Bm$ is ideal if and only if $\mathscr{P}(\Cc, \yv)$ has only integral extreme points. 
The widely-studied balanced and totally unimodular matrices (TUM) are special cases of perfect and ideal matrices. Fig. \ref{fig:matrices} presents their relations.
\begin{figure}[htb]
\centering
\includegraphics[width=0.6\columnwidth]{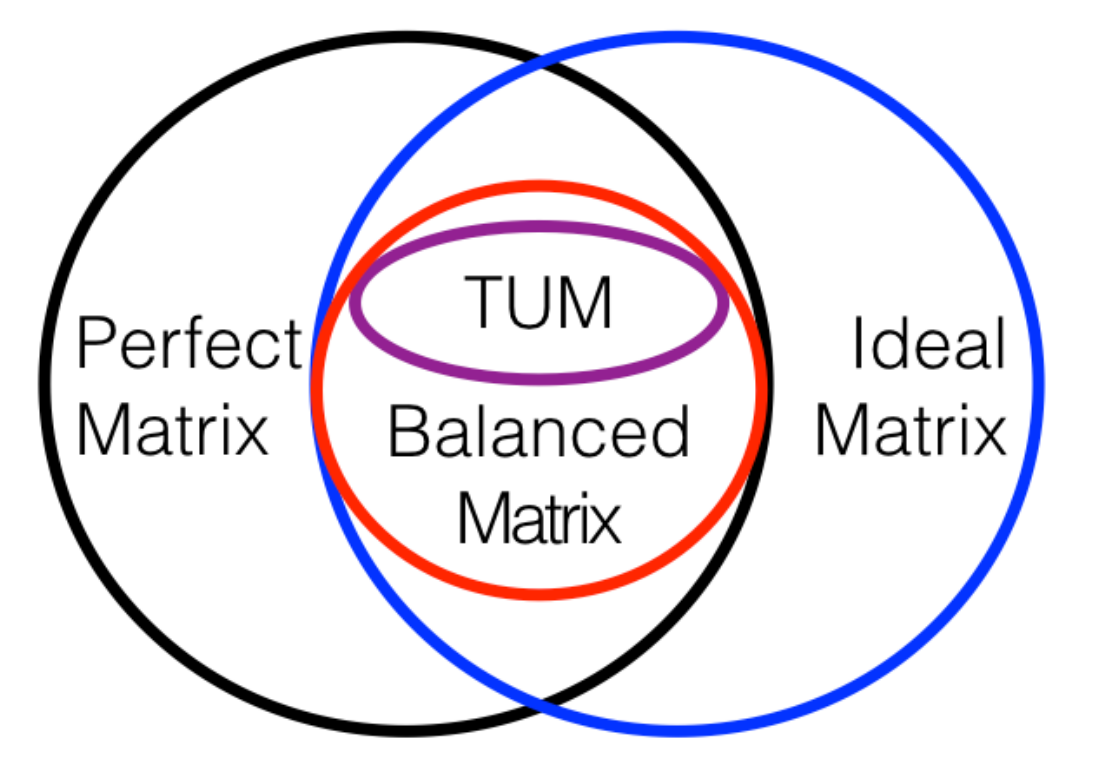}
\caption{The relation among perfect, ideal, balanced, and totally unimodular matrices (TUM).}
\label{fig:matrices}
\end{figure}

\subsection{The Optimality of Orthogonal Access}
By the above outer bounds, we identify three families of network topologies for which orthogonal access achieves the optimal DoF region of TIM-MP problems.
The conditions of the optimality of orthogonal access are summarized in Theorem \ref{theorem:OA-optimal}, and will be detailed case by case in the ensuing theorems.
\begin{theorem}[Optimality of Orthogonal Access]
\label{theorem:OA-optimal}
For the TIM-MP problem with the conflict digraph $\Dc$, the clique-vertex incidence matrix $\Am$ and the dicycle-vertex incidence matrix $\Bm$, orthogonal access via fractional acyclic set coloring achieves the optimal DoF region, if it falls in any one of the following cases.
\begin{itemize}
\item {\rm Case I}: The conflict digraph $\Dc$ contains no dicycles $C_n$ with $n \ge 3$, and $\Am$ is a perfect matrix;
\item {\rm Case II}: The conflict digraph $\Dc$ contains no cliques $Q_n$ with $n \ge 3$, and $\Bm$ is an ideal matrix;
\item {\rm Case III}: The symmetric part $S(\Dc)$ is a perfect graph, and the dicycle-vertex incidence matrix $\Bm'$ of the remaining conflict digraph $\Dc$ after removing cliques $Q_n$ $(n \ge 3)$ contains no minimally non-ideal submatrices.
\end{itemize}
\end{theorem}

The converse proof relies on the integrality of set packing and covering polytopes, which has been established in polyhedral combinatorics. The achievability is due to acyclic set coloring. The sub-digraphs in the conflict digraph induced by the coordinates of the extreme points of these polytopes are acyclic sets. The detailed proofs will be shown case by case in the ensuing theorems.

For the TIM setting, it has been shown in \cite[Theorem~1]{Yi:Fractional} that orthogonal access achieves the all-unicast DoF region of the TIM problem if and only if the network topology is chordal bipartite.
Analogously, we have the following theorem for the TIM-MP problem when message passing is enabled. %
\begin{theorem}[Case I]
\label{theorem:perfect}
For the family of networks in which conflict digraphs $\Dc$ contain no dicycles $C_n$ with $n \ge 3$, and $\Am$ is a perfect matrix, the optimal DoF region achieved by orthogonal access can be characterized by the set packing polytope
\begin{align} \label{eq:case-I-region}
\mathscr{D}=\left\{(d_{\Kc}): \Pmatrix{0 \le d_k \le 1, \forall k \in \Kc \\ \sum_{k \in Q} d_i \le 1, \; \forall Q \in \Qc(S(\Dc))} \right\}
\end{align}
where the undirected graph $S(\Dc)$ is the symmetric part of the conflict digraphs $\Dc$, and $\Qc(S(\Dc))$ is the set of all maximal cliques in $S(\Dc)$.
\end{theorem}
\begin{proof}
See Appendix \ref{proof:perfect}.
\end{proof}

\begin{remark}
\normalfont
The condition that conflict digraph $\Dc$ contains no dicycles $C_n$ with $n \ge 3$, and $\Am$ is a perfect matrix, indicates that $\Dc$ is a perfect digraph \cite{combinopt}. According to the definition of perfect digraphs in Appendix \ref{appendix:graph}, Theorem \ref{theorem:perfect} identifies the optimality of orthogonal access for the conflict digraph $\Dc$ {\em excluding} the following cases:
\begin{itemize}
\item $\Dc$ contains dicycles $C_n$ with length $n\ge 3$ as induced sub-digraph;
\item $\Dc$ contains filled odd holes or filled odd antiholes, i.e., its symmetric part $S(\Dc)$ contains odd holes or odd antiholes.
\end{itemize}

Similarly to \cite{Yi:Fractional}, the characterization of the optimal DoF region automatically yields the optimality of the traditional metrics such as sum or symmetric DoF.
\hfill $\lozenge$
\end{remark}

\begin{remark}
\normalfont
For a perfect digraph $\Dc$, acyclic set coloring of $\Dc$ reduces to vertex coloring of its symmetric part $S(\Dc)$. Thus, any feasible coloring of the symmetric graph $S(\Dc)$ is also feasible for $\Dc$ \cite{SPDT}.
If the conflict digraph $\Dc$ only has symmetric part, then $\chi_A(\Dc)=\chi(S(\Dc))$. As such, the orthogonal access of our problem is reduced to that of the TIM problem without message passing, because in this case interference is mutual, and message passing does not help.
\hfill $\lozenge$
\end{remark}

\begin{example}
\normalfont
Consider a 6-cell network topology shown in Fig. \ref{fig:thm1-ex1}(a). In the conflict digraph $\Dc$ in Fig. \ref{fig:thm1-ex1}(b), the symmetric part $S(\Dc)$ in Fig. \ref{fig:thm1-ex1}(c) is perfect, and there do not exist dicycles $C_n$ with $n \ge 3$ as induced sub-digraph, although there exist dicycles for instance $\{1,2,6\}$. As there is an arc $(6,2)$, the sub-digraph induced by $\{1,2,6\}$ is not a dicycle. Thus, according to Theorem \ref{theorem:perfect}, we have the optimal DoF region
\begin{align*}
\mathscr{D}=\left\{ (d_1,\dots,d_6) \in \mathbb{R}_+^6 : \Pmatrix{d_1 \le 1, d_3 \le 1, d_5 \le 1, \\ d_2+d_4+d_6 \le 1.} \right\}.
\end{align*}
It immediately follows that the symmetric and sum DoF are $d_\sym=\frac{1}{3}$ and $d_{{\rm sum}}=4$, respectively. To achieve the symmetric DoF of $\frac{1}{3}$, we can simply schedule $\{W_1,W_2\}$, $\{W_3,W_4\}$, and $\{W_5,W_6\}$ in three time slots respectively. In each time slot, $W_1,W_3,W_5$ are free of interference, and $W_2,W_4,W_6$ can be subsequently decoded after passing the decoded messages $W_1,W_3,W_5$ at receivers 1, 3, 5 to receivers 2, 4, 6 respectively. \hfill $\lozenge$
\begin{figure}[htb]
 \centering
\includegraphics[width=0.9\columnwidth]{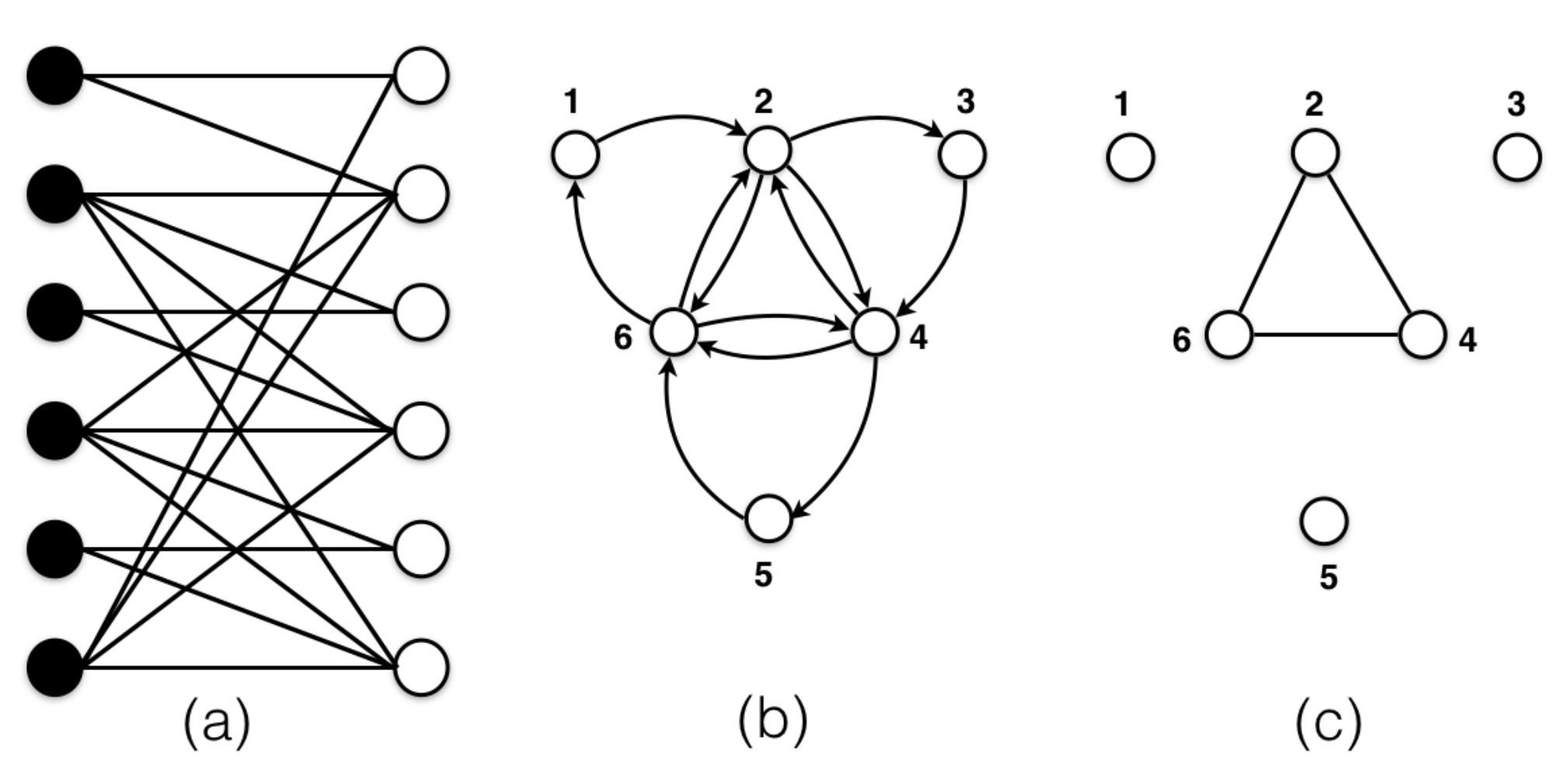}
\caption{(a) A network topology, (b) its conflict digraph $\Dc$ that is a perfect digraph, and (c) its (undirected) symmetric part $S(\Dc)$ that is also a perfect (undirected) graph.}
\label{fig:thm1-ex1}
\end{figure}
\end{example}

Note however that, the condition that the conflict digraph is perfect in Theorem \ref{theorem:perfect} is only sufficient but not necessary. A counterexample is the dicycle $C_3$, where the fractional acyclic set coloring achieves the DoF region 
\begin{align}
\mathscr{D}=\left\{ (d_1,d_2,d_3): \Pmatrix{0 \le d_k \le 1, \forall k \\ d_1+d_2+d_3 \le 2} \right\}
\end{align}
while the conflict digraph is not perfect. 

Realizing that Case I focuses only on the integrality of the set packing polytopes, we identify another family of networks focusing on the integrality of the set covering polytopes in the following theorem, where orthogonal access is still DoF-optimal albeit the conflict digraph is not perfect. 

\begin{theorem}[Case II]
\label{theorem:ideal}
For the family of networks in which $\Dc$ contains no cliques $Q_n$ with $n \ge 3$, and $\Bm$ is an ideal matrix, the optimal DoF region achieved by orthogonal access can be given by the set covering polytope
\begin{align} \label{eq:dof-cycle}
\mathscr{D}=\left\{(d_{\Kc}): \Pmatrix{0 \le d_k \le 1, \forall k \in \Kc \\ \sum_{k \in C} d_k \le \abs{C}-1, \forall C \in \Cc} \right\}
\end{align}
where $\Cc$ is the collection of all dicycles without chord.
\end{theorem}

\begin{proof}
See Appendix \ref{proof:ideal}.
\end{proof}

\begin{remark}
\normalfont
As a matter of fact, the condition that $\Bm$ is ideal already {\em excludes} the existence of the cliques of size 3 or more in the conflict digraph. It is because otherwise a clique of size 3 or more will lead to the existence of a non-ideal circulant submatrix $\Cm_3^2$ of $\Bm$ which results in a contradiction to that any submatrix (minor) of an ideal matrix is also ideal. 
The definition of the circulant matrix can be found in Appendix \ref{appendix:polyhedral}.
The circulant matrices that are ideal consists only of $\Cm_{n}^2$ for even $n \ge 4$, $\Cm_{6}^3$, $\Cm_9^3$, and $\Cm_8^4$. 
\hfill $\lozenge$
\end{remark}

\begin{remark}
\normalfont
Differently from the perfect matrices, it is still an open problem to fully characterize all the ideal matrices.
It has been shown in \cite{padberg1993lehman} that a matrix is ideal if and only if it does not contain a minimally non-ideal (MNI) submatrix minor (see definitions in Appendix \ref{appendix:polyhedral}). The MNI matrices are the ``smallest'' possible matrices that are not ideal \cite{lehman1979width,lehman1990width}. A submatrix of $\Mm$ is a minor of $\Mm$ if it can be obtained from $\Mm$ by successively deleting a column $j$ and the rows with a `1' in column $j$. If a matrix is ideal then so are all its minors. A matrix $\Mm$ is MNI, if it is not ideal but all its proper minors are.
For instance, $\Cm_3^2$ is a MNI matrix, so any matrix that contains $\Cm_3^2$ as a minor is not ideal, e.g., Fig. \ref{fig:ideal}(c). 
\end{remark}

\begin{remark}
\normalfont
As special cases of ideal matrices, balanced and totally unimodular (TU) matrices are completely characterizable.  \footnote{Note here that, although TU and balanced matrices are special cases of perfect matrices, the conflict digraphs with incidence matrices being TU or balanced are not the subclass of those in Theorem \ref{theorem:perfect}, because two different incidence matrices are considered.}
The characterization of balanced or totally unimodular matrices is well understood.
A matrix is balanced if and only if it contains no odd hole matrices (i.e., $\Cm_n^2$ with odd $n$ where $n \ge 3$) as submatrices. A polynomial time recognition algorithm for balanced matrices was also given with the aid of decomposition \cite{conforti1999decomposition}. 
The full characterization of all TU matrices was also given in \cite{seymour1980decomposition}, where a matrix is TU if and only if it is a certain natural combination of some matrices and some copies of a particular 5-by-5 TU matrix.
\hfill $\lozenge$
\end{remark}

\begin{example}
\normalfont
Firstly, consider a 4-user TIM-MP instance, as shown in Fig. \ref{fig:ideal}(a). There are two dicycles $\{(1,2),(2,3),(3,1)\}$ and $\{(1,4),(4,3),(3,1)\}$ without cliques, such that the outer bund is given by two cycle inequalities $d_1+d_2+d_3 \le 2$ and $d_1+d_3+d_4 \le 2$ as well as the individual ones $d_k \le 1$, $\forall k$.
It is not hard to verify that the dicycle-vertex incidence matrix 
\begin{align} \label{eq:mni}
\Bm=\Bmatrix{1 & 1 & 1 & 0 \\ 1 & 0 & 1 & 1}
\end{align}
 is totally unimodular, and thus ideal. The extreme points of the polytope consist of all possible 4-tuples excluding $(1,0,1,1)$, $(1,1,1,0)$ and $(1,1,1,1)$. It can be checked that all these extreme points can be achieved by acyclic set coloring. As such, the DoF region can be achieved by time sharing among these extreme points. 
 \begin{figure}[htb]
 \centering
\includegraphics[width=0.8\columnwidth]{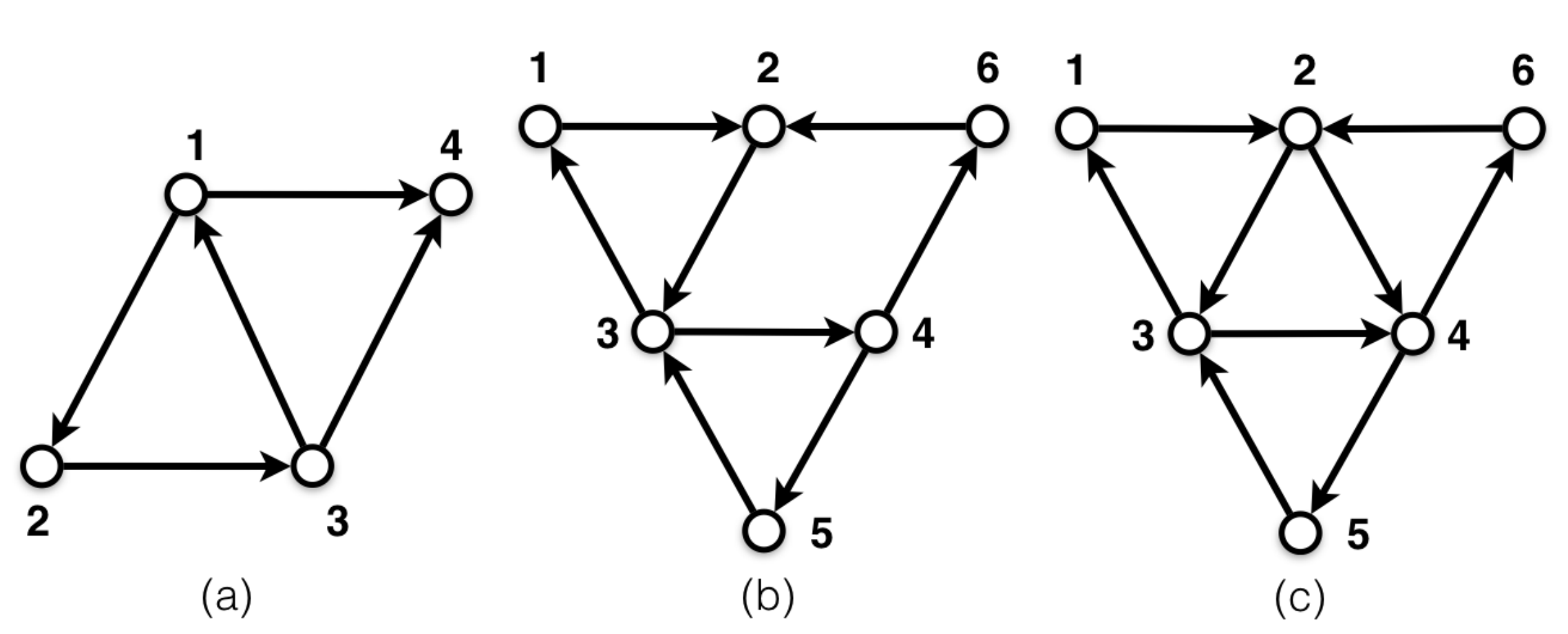}
\caption{The conflict digraphs $\Dc$ with (a) a totally unimodular, (2) an ideal and (c) a nonideal cycle-vertex incidence matrices.}
\label{fig:ideal}
\end{figure}
 
For the instance in Fig. \ref{fig:ideal}(b), the corresponding dicycle-vertex incidence matrix is ideal, while in Fig. \ref{fig:ideal}(c) with an arc $(2,4)$ added, the resulting dicycle-vertex incidence matrix is not ideal any more. In Fig. \ref{fig:ideal}(c), the corresponding dicycle-vertex incidence matrix is
\begin{align}
\Bm=\begin{bmatrix} 1 & 1 & 1 & 0 & 0 & 0 \\ 0 & 0&1&1&1&0\\0&1&0&1&0&1 \end{bmatrix}
\end{align}
which contains an MNI matrix $\Cm_3^2$ as a submatrix. So, it is not a balanced matrix, nor an ideal matrix.
 \hfill $\lozenge$
\end{example}

In the following corollary, we give some explicit characterization on conflict digraphs when orthogonal access is DoF optimal, according to Theorem \ref{theorem:ideal}. The proof is relegated to Appendix \ref{proof:corollary-theorem-4}.
\begin{corollary}
\label{corollary:theorem-4}
For the TIM-MP problems, orthogonal access achieves the optimal DoF region given by \eqref{eq:dof-cycle}, if any one of the following conditions is satisfied.
\begin{itemize}
\item All the induced dicycles in conflict digraph $\Dc$ are disjoint (i.e., none of two induced dicycles share vertices).
\item There exist at most two chordless dicycles in conflict digraph $\Dc$, including $(K,2)$ regular network whose conflict digraph is a single dicycle.
\end{itemize}
\end{corollary}

By Theorems \ref{theorem:perfect} and \ref{theorem:ideal}, we show the optimality of orthogonal access for the three-user network with all possible topologies (16 non-isomorphic instances in total), whose proof is relegated to Appendix \ref{proof:cor-3-user}.
\begin{corollary}
\label{cor:3-user}
For the three-user TIM-MP problem, orthogonal access achieves the optimal DoF region.
\end{corollary}

Theorem \ref{theorem:perfect} handles the perfect conflict digraphs with neither dicycles of length 3 or more, nor odd hole or antihole in the symmetric part, while Theorem \ref{theorem:ideal} considers the case with dicycles but without clique of size 3 or more. The former only considers the tightness of clique inequalities, whereas the latter focuses only on the cycle inequalities. Beyond Theorems \ref{theorem:perfect} and \ref{theorem:ideal}, we come up with anther sufficient condition on the optimality of orthogonal access, where both dicycles and cliques are contained in the conflict digraphs.
\begin{theorem}[Case III]
\label{theorem:balanced}
For the family of networks in which $S(\Dc)$ is a perfect graph, and the dicycle-vertex incidence matrix $\Bm'$ of the remaining conflict digraph $\Dc$ after removing cliques $Q_n$ $(n \ge 3)$ contains no MNI submatrices, the optimal DoF region achieved by orthogonal access can be given by the mixed set packing and covering polytope
\begin{align}
\mathscr{D}=\left\{(d_{\Kc}): \Pmatrix{0 \le d_k \le 1, \forall k \in \Kc \\ \sum_{k \in C} d_k \le \abs{C}-1, \forall C \in \Cc' \\ \sum_{k \in Q} d_k \le 1, \forall Q \in \Qc'} \right\}
\end{align}
where $\Cc'$ is the collection of all minimal dicycles and $\Qc'$ is the collection of all maximal cliques with size no less than 3, and
$
\abs{C \cap Q} \le 1, \forall C \in \Cc', \forall Q \in \Qc'.
$
\end{theorem}
\begin{proof}
See Appendix \ref{proof:balanced}.
\end{proof}

\begin{remark}
\normalfont
A special type of networks in case III, is that the concatenation of clique-/dicycle-incidence matrices $\left[\begin{smallmatrix} \Am \\ \Bm \end{smallmatrix}\right]$ is balanced, so that the polytope is integral for any integer $\abs{C}$. 
\end{remark}

\begin{remark}
\normalfont
It is still an open problem to fully sort out all the MNI matrices, while it has been proven that the MNI matrices do have some properties. It has been characterized by Lehman \cite{lehman1979width,lehman1990width} that if $\Mm$ is an MNI matrix, then $\Mm$ is isomorphic (up to a permutation of rows followed by a permutation of columns) to either (1) the degenerate projective plane $\Jm_n$ with $n \ge 2$, or (2) $\Mm=\left[\begin{smallmatrix} \Mm_1\\ \Mm_2 \end{smallmatrix} \right]$ where $\Mm_1$ is a square nonsingular matrix with $r\ge 2$ `1's per row and per column, and each row of $\Mm_2$ has at least $(r+1)$ `1's. The known MNI matrices include the circulant matrices $\Cm_n^2$ for odd $n \ge 3$, $\Cm_5^3$, $\Cm_8^3$, $\Cm_{11}^3$, $\Cm_{14}^3$, $\Cm_{17}^3$, $\Cm_7^4$, $\Cm_{11}^4$, $\Cm_9^5$, $\Cm_{11}^6$, and $\Cm_{13}^7$ \cite{cornuejols1994ideal},  the degenerate projective planes $\Jm_n$ with $n \ge 3$, and the Fano plane $\Fm_7$. %
Fig. \ref{fig:mni} presents some conflict digraphs whose dicycle-vertex incidence matrices are MNI matrices.
\begin{figure}[htb]
\centering
\includegraphics[width=1\columnwidth]{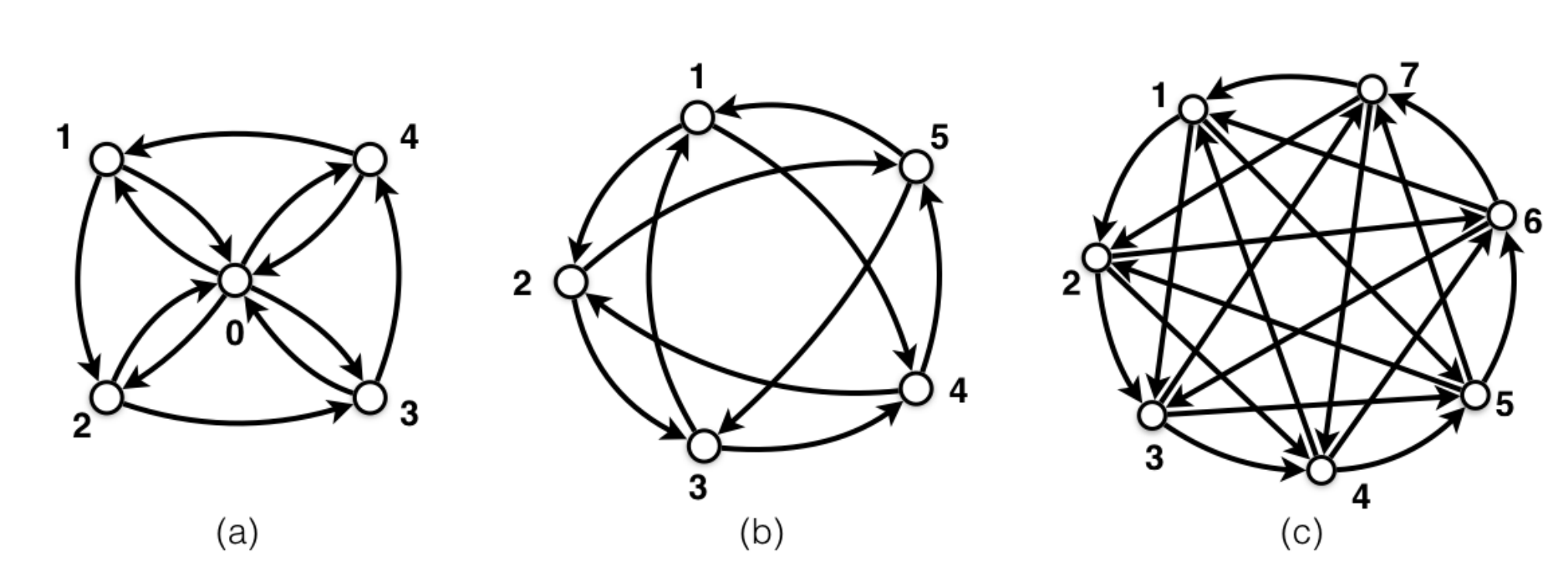}
\caption{Conflict digraphs $\Dc$ whose dicycle-vertex incidence matrices are MNI matrices: (a) the degenerate projective plane $\Jm_4$, (b) the circulant matrix $\Cm_5^3$, and (c) the Fano plane $\Fm_7$.}
\label{fig:mni}
\end{figure}
\end{remark}

\begin{example}
\normalfont
\begin{figure}[htb]
\centering
\includegraphics[width=0.8\columnwidth]{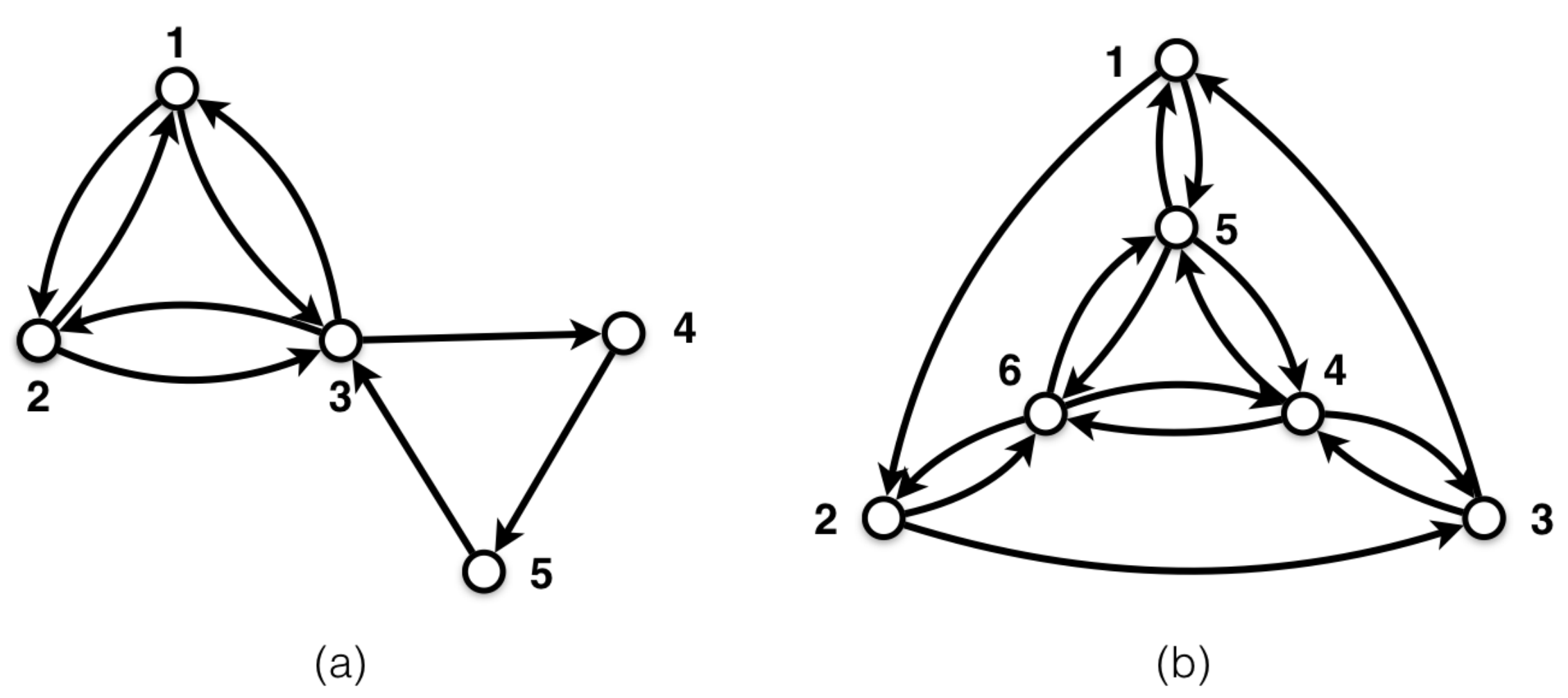}
\caption{Two conflict digraphs $\Dc$ with both clique of size 3 and dicycles, and satisfy condition in Case III.}
\label{fig:case3}
\end{figure}
 
For the instance in Fig. \ref{fig:case3}(a), the optimal DoF region is given by 
\begin{align}
\mathscr{D}=\left\{(d_k):\Pmatrix{0 \le d_k \le 1, \forall k\\
d_1+d_2+d_3 \le 1\\
d_3+d_4+d_5 \le 2} \right\}
\end{align}
where the cycle inequalities $d_i+d_j \le 1$ for $i\ne j \in \{1,2,3\}$ are replaced by a single clique inequality. The extreme points of $\mathscr{D}$ are 5-tuple binary vectors apart from $(1,1,*,*,*)$, $(1,*,1,*,*)$, $(*,1,1,*,*)$ and $(*,*,1,1,1)$, where $*$ denotes either 0 or 1. It can be easily checked that all extreme points are achievable by acyclic set coloring.

For the instance in Fig. \ref{fig:case3}(b), the optimal DoF region is given by 
\begin{align}
\mathscr{D}=\left\{(d_k):\Pmatrix{0 \le d_k \le 1, \forall k, d_1+d_5 \le 1\\
 d_2+d_6 \le 1, \; d_3+d_4 \le 1\\
d_1+d_2+d_3 \le 2\\
d_4+d_5+d_6 \le 1} \right\}
\end{align}
where the cycle inequalities $d_i+d_j \le 1$ for $i\ne j \in \{4,5,6\}$ are replaced by a single clique inequality. After such a replacement, the resulting polytope is integral, and the extreme points are achievable by acyclic set coloring.
 \hfill $\lozenge$
\end{example}

\section{A Generalized Index Coding Problem}
\label{sec:sic}
Building upon the relation between index coding and TIM \cite{Jafar:2013TIM}, we also establish an analogous relation to TIM-MP.
As in Appendix \ref{appendix:index}, the goal of index coding is to minimize the number of transmissions such that all receivers are able to decode their own messages {\em simultaneously}. 
As TIM-MP is a generalization of TIM, we introduce a generalization of index coding, referred to as ``successive index coding (SIC)'', where the message decoding is not necessarily simultaneous.

\subsection{Successive Index Coding (SIC)}
In the SIC problem, the receivers are allowed to declare their own messages once they decode them. Such a declaration offers other receivers additional side information, by which the minimal number of transmissions can be further reduced.  

The multiple-unicast SIC problem considers a noiseless broadcast channel, where a transmitter wants to send the message $W_j$ to the receiver $j$ who has access to prior knowledge of initial side information $W_{\Sc_j}$ ($\Sc_j \subseteq \Kc \backslash \{j\}$) as well as the additional side information due to the successive decoding and message passing. The goal is to find out the minimum number of transmissions (i.e., broadcast rate) over all possible successive decoding and message passing orders such that each receiver can successively decode its desired message.

The additional side information depends on the decoding order. For the single-round message passing, given a decoding order $\pi$, the partial order $i \prec_{\pi} j$ indicates the message passing from  receivers $i$ to $j$, which is equivalent to enhancing the side information set $\Sc^{\pi}_{j} \gets \Sc_{j} \cup W_{i}$. As such, the enhanced side information set can be written by
\begin{align} \label{eq:updateinfo}
\Sc^{\pi}_{j} = \Sc_{j} \bigcup_{i: i \prec_{\pi} j} W_{i}
\end{align}
for a specific decoding order $\pi$. 

In what follows, we formally define the $n$-receiver SIC problem.
A $(t_1,\dots,t_n,r)$ successive index coding scheme with side information index sets $\{\Sc_1,\dots,\Sc_n\}$ and a given decoding order $\pi$ consists of the following:
\begin{itemize}
\item an encoding function, $\phi: \prod_{i=1}^n \{0,1\}^{t_i} \mapsto \{0,1\}^r$ at the transmitter that encodes $n$-tuple of messages (codewords) $x^n$ to a length-$r$ index code.
\item a decoding function at the receiver $j$, $\psi_j^{\pi}: \{0,1\}^r \times \prod_{k \in \Sc_j^{\pi}} \{0,1\}^{t_k} \mapsto \{0,1\}^{t_j}, \; \forall j$ that decodes the received index code back to $x_j$ with initial side information $x(\Sc_j)$ held at receiver $j$ as well as the passed messages that decoded earlier $\bigcup_{i: i \prec_{\pi} j} W_{i}$.
\end{itemize}

The {\em initial} side information digraph $\bar{\Dc}$ or conflict digraph $\Dc$ of a successive index coding instance is identical to that of the corresponding index coding instance.

Thus, for a given decoding order $\pi$, a rate tuple $(R_1^{\pi},\dots,R_n^{\pi})$ is said to be achievable if there exists a successive index code $(t_1,\dots,t_n,r)$ with
\begin{align}
\psi_j^{\pi}(\phi(x^n),x(\Sc_j^{\pi})) = x_j, \quad \forall j
\end{align}
such that any rate tuple $(R_1^{\pi},\dots,R_K^{\pi})$ in the rate region $\Rc^{\pi}$ is achievable with
\begin{align}
R_j^{\pi} \le \frac{t_j}{r}, \quad \forall j.
\end{align}
Similarly to the TIM-MP problem, the capacity region $\mathscr{C}$ of the SIC problem is the set of all achievable rate tuples where time sharing among multiple single-round message passing is allowed. More specifically, it is the convex hull of the union of the achievable rate regions for all possible decoding orders, i.e., $\mathscr{C} = {\rm conv} (\cup_{\pi} \Rc^{\pi})$.

By the channel enhancement approach in \cite{Jafar:2013TIM}, it is readily shown that the DoF region of every TIM-MP instance is outer bounded by the capacity region of the corresponding SIC instance, and both problems are equivalent for linear coding schemes (i.e., with linear encoding/decoding functions). In particular, for each single-round message passing, given a decoding order $\pi$, the resulting TIM-MP (SIC) problem can be treated as a modified TIM (index coding) problem with updated side information set in \eqref{eq:updateinfo}.

As the linear coding schemes are considered in the previous sections, the results obtained for TIM-MP are applicable to SIC. Specifically, the achievable symmetric DoF or DoF region of TIM-MP with conflict digraph $\Dc$ are also the achievable symmetric rate or rate region of SIC with initial side information digraph $\bar{\Dc}$. Similarly, the sufficient or necessary conditions seen before for orthogonal access in TIM-MP are also applicable to the corresponding SIC setup. 

In the rest of this section, we will focus on the broadcast rate, defined as
\begin{align}
\beta^{\rm SIC}({\Dc}) = \inf_{b:(\frac{1}{b},\dots,\frac{1}{b}) \in \mathscr{C}} b,
\end{align}  
which is the minimum number of transmissions (i.e.,  the number of transmitted symbols over the shared link normalized by the total message length) for the SIC problem, also known as reciprocal symmetric capacity.

\begin{figure*}[btp]
\begin{center}
\begin{tabular*}{0.78\textwidth}{  c || c | c  }
  \hline \hline
  Achievability & Index Coding & Successive Index Coding \\ \hline
  Clique Covering & $\beta \le \chi(U(\Dc))$ & $\beta^{\sic} \le \chi_A(\Dc)$ \\
  \hline 
   Fractional Clique Covering & $\beta \le \chi_f(U(\Dc))$ & $\beta^{\sic} \le \chi_{A,f}(\Dc)$ \\
  \hline 
  Cycle Covering  & $\beta \le \sum_{C \in \Cc(\bar{\Dc})} (\abs{C}-1)$  & $\beta^{\sic} \le \sum_{C \in \Cc(\Dc)} \frac{\abs{C}}{\abs{C}-1}$   \\ \hline
  Partial Clique Covering & $\beta \le \min\limits_{\Vc=\{\Vc_1,\dots,\Vc_s\}} \sum\limits_{i=1}^s (k_i+1)$ & $\beta^{\sic} \le \min\limits_{\Vc=\{\Vc_1,\dots,\Vc_s\}} \sum\limits_{i=1}^s (m_i+1)$ \\
  \hline \hline
\end{tabular*}
\end{center}
\caption{The analogous achievability of index coding and successive index coding, where $U(\Dc)$ is the underlying undirected graph of conflict digraph $\Dc$, $\Cc(\bar{\Dc})$ and $\Cc({\Dc})$ are collections of disjoint dicycles in side information digraph $\bar{\Dc}$ and conflict digraph $\Dc$, respectively, and $m_i \le k_i =\Delta^-(\Dc[\Vc_i])$ shows gains of successive index coding over index coding.}
\label{fig:analogy}
\end{figure*}

\subsection{Analogy to Index Coding}
Analogously to the index coding problem, we define some achievability schemes for the SIC problem.
As partial clique covering is a generalized version of clique and cycle covering, we only present a definition of partial cliques in conflict digraphs $\Dc$ of SIC instances, named weakly degenerate set.
\begin{definition}[Weakly Degenerate Set]
A conflict sub-digraph ${\Dc}[\Qc]$ is a weakly $m$-degenerate set if any of its induced sub-digraph has a vertex of out-degree or in-degree no more than $m$, i.e., for all $\Qc' \subseteq \Qc$, $\exists v_q \in \Qc'$, $\min \{d^-({\Dc}[\Qc'],v_q),d^+({\Dc}[\Qc'],v_q)\} \le m$.
\end{definition}
\begin{remark}
\normalfont
The weakly degenerate set is a generalized version of partial clique to the SIC problem.
The acyclic set is weakly 0-degenerate, the dicycle is weakly 1-degenerate, and the clique $Q_k$ is weakly $(k-1)$-degenerate.
\end{remark}
\begin{lemma}[Weakly Degenerate Set Covering]
The broadcast rate of the SIC problem with conflict digraph ${\Dc}=(\Vc,{\Ac})$ is upper bounded by
\begin{align}
\beta^{\sic}(\Dc) \le \min_{\{\Vc_1, \dots, \Vc_s\}} \sum_{i=1}^s (m_i +1),
\end{align}
with weakly degenerate set covering, where the minimum is over all partitions of $\Vc=\{\Vc_1, \dots, \Vc_s\}$, and for all $i = 1,\dots,s$, ${\Dc}[\Vc_i]$ is a weakly $m_i$-degenerate set.
\end{lemma} 

The analogy of achievability between index coding and successive index coding problems is summarized in Fig. \ref{fig:analogy}. Note here that, for index coding, (fractional) clique covering on the side information digraph is equivalent to (fractional) vertex coloring on the underlying undirected conflict graph.

As a generalized version of acyclic set, however, weakly generate set covering does not offer any improvement over the acyclic set coloring, unlike that in the index coding problem where partial clique covering indeed outperforms clique covering (i.e., independent set coloring on conflict graphs). It is because $\chi_A(\Dc[\Qc]) \le m+1$ if $\Dc[\Qc]$ is a weakly $m$-degenerate set, meaning that weakly degenerate set covering does not offer gains over acyclic set coloring. In other words, message passing renders the (weakly degenerate) set partition of conflict digraphs useless for SIC problems in terms of broadcast rate. Nevertheless, the weakly degenerate set partition has the potential to restrict message passing locally, shortens the decoding latency of the entire network, and facilitates the tradeoff between broadcast rate and message passing overhead.

\subsection{Reducibility and Criticality}
\subsubsection{Reducibility: When is a Vertex/Message Reducible?}
The objective of studying vertex-reducibility is to remove the vertices in the conflict graph without changing the broadcast rate, so as to reduce the large-size SIC instance to a smaller one with less vertices.
\begin{definition}[Reducibility]
A vertex $v$ in the conflict digraph $\Dc$ is reducible, if its removal does not decrease the broadcast rate of the corresponding SIC problem, i.e., $\beta^{\rm SIC} (\Dc-v)=\beta^{\rm SIC} (\Dc)$.
\end{definition}

By strong component decomposition (see Appendix \ref{appendix:graph}), we have the following theorem.
\begin{theorem} \label{theorem:reducibility}
Given the unique strong component decomposition $\Dc=\{\Dc_1,\dots,\Dc_s\}$, let $\Dc^*_s$ be the strong component with the maximal fractional dichromatic number. If $\Dc^*_s$ falls into the digraph classes in Theorem \ref{theorem:OA-optimal}, then the vertices in $\Vc(\Dc) \backslash \Vc(\Dc_s^*)$ are reducible. 
\end{theorem}
\begin{proof}
See Appendix \ref{proof:reducibility}.
\end{proof}
\begin{remark}
\normalfont
The vertices that are not involved in any dicycles (e.g., with either only incoming or outgoing arcs) are reducible. Any vertex in a directed acyclic graph is reducible. This agrees with the fact that the broadcast rate of the SIC instances with conflict digraphs being directed acyclic is 1. 
\hfill $\lozenge$
\end{remark}
\begin{remark}
\normalfont
The strong component with the maximal fractional chromatic number is not necessarily the one with the largest size.
For the strong component with the largest size, denoted by $\Dc[\Vc_s^*]$, the index coding problem corresponding to $\Dc[\Vc_s^*]$ serves an upper bound of the original successive index coding problem, i.e., $\beta^{\rm SIC}(\Dc) \le \beta(\Dc[\Vc_s^*])$. That is, if a broadcast rate of the index coding instance is achievable, it is also achievable for the original SIC instance. It is because choosing one single vertex from each strong component forms an acyclic set, and thus the broadcast rate of $\Dc[\Vc_s^*]$ for the index coding setting without message passing dominates. \hfill $\lozenge$
\end{remark}
\begin{example}
\normalfont
Let us take two conflict digraphs shown in Fig. \ref{fig:strong} as illustrative examples. The induced sub-digraphs in the shadow are the strong components with the maximal fractional dichromatic numbers. The strong components in shadow are a clique $Q_4$ (on the left) and a dicycle $C_2$, both of which are perfect digraphs and fall into the cases in Theorem \ref{theorem:OA-optimal}. Thus, both SIC instances can be reduced without loss of broadcast rate to the ones with only the strong components in the shadow, so that $\beta^{\rm SIC}(\Dc)=\beta^{\rm SIC}(\Dc_s^*)=4$ for Fig. \ref{fig:strong}(a) and $\beta^{\rm SIC}(\Dc) = \beta^{\sic}(\Dc_s^*)=2$ for Fig. \ref{fig:strong}(b). Note also in Fig. \ref{fig:strong}(b) that the dicycle $C_3$ is the strong component with the largest size, but its dichromatic number $\chi_{A,f}(\Dc[\Vc_s^*])=\frac{3}{2}$ is not maximal, because the strong component in the shadow has $\chi_{A,f}(\Dc_s^*)=2$. 
Thus, we have an upper bound $\beta^{\rm SIC}(\Dc) \le \beta(\Dc[\Vc_s^*])=2$. Together with the cycle bound $\beta^{\rm SIC}(\Dc) \ge 2$, we also have the optimal broadcast rate $\beta^{\rm SIC}(\Dc)=2$.
\hfill $\lozenge$
\begin{figure}[htb]
 \centering
\includegraphics[width=0.8\columnwidth]{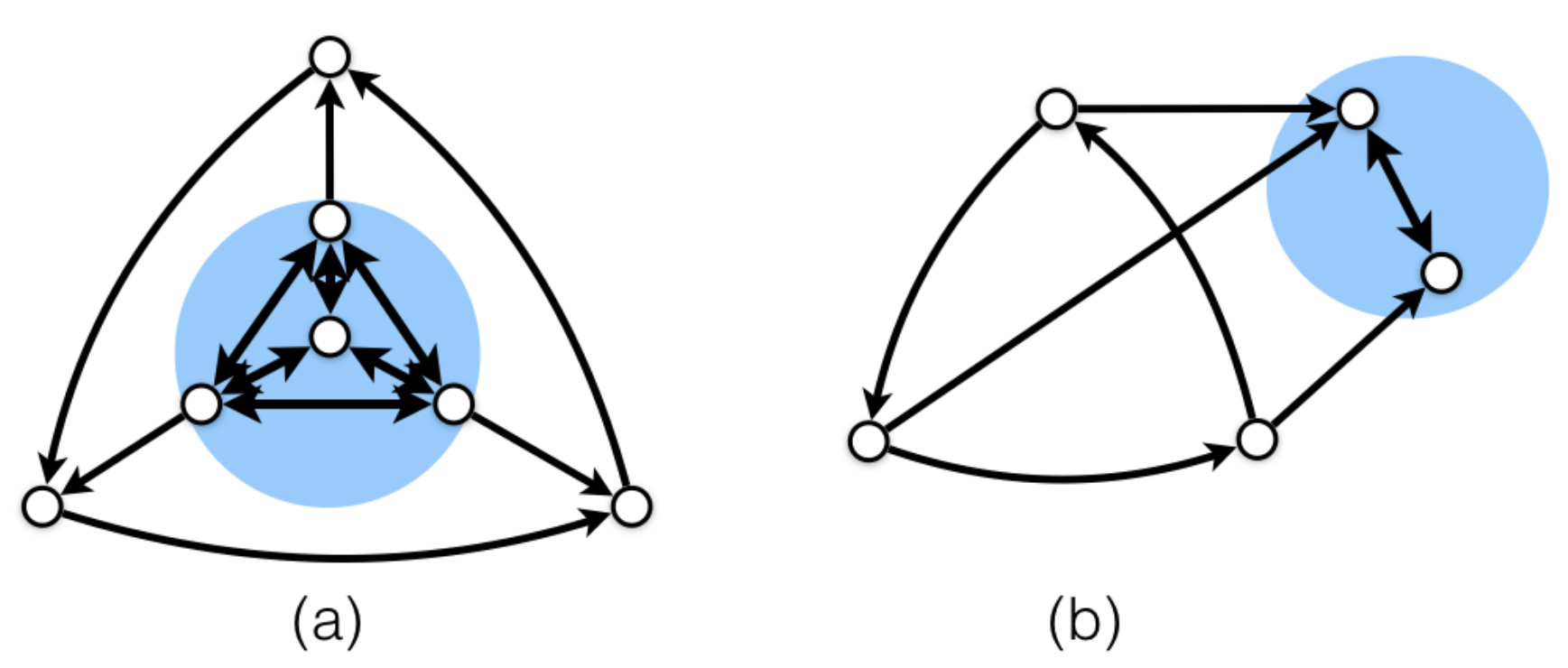}
\caption{The conflict digraphs $\Dc$ with the strong components where the parts in the shadow are the strong components with the maximal dichromatic numbers. The broadcast rates of these SIC instances are $\beta^{\rm SIC}(\Dc)=\beta^{\rm SIC}(\Dc_s^*)=4$ and $\beta^{\rm SIC}(\Dc)=\beta^{\rm SIC}(\Dc_s^*)=2$, respectively. The vertices lying outside the shadow are reducible.}
\label{fig:strong}
\end{figure}
\end{example}

\subsubsection{Criticality: When is an Arc/Interference Critical?}
The objective of studying arc-criticality is to remove the arcs in the conflict graph without changing the broadcast rate, so as to reduce the SIC instance with a large arc set to a smaller one with less arcs.
\begin{definition}[Criticality]
An arc $e$ in the conflict digraph $\Dc$ is critical, if its removal strictly decreases the broadcast rate of the corresponding SIC problem, i.e., $\beta^{\rm SIC} (\Dc-e) < \beta^{\rm SIC} (\Dc)$.
\end{definition}
The removal of an arc (i.e., interference link) does not increase the interference in the network, so the broadcast rate should not be increased. A conflict digraph $\Dc$ is said to be critical, if every arc in $\Dc$ is critical. %

\begin{theorem} \label{theorem:criticality}
In a conflict digraph $\Dc$, an arc is critical, if it belongs to the following:
\begin{itemize}
\item the unique minimal dicycle when the dicycle-vertex incidence matrix of $\Dc$ is ideal;
\item the unique maximal clique when $\Dc$ is a perfect digraph.
\end{itemize}
\end{theorem}
\begin{proof}
See Appendix \ref{proof:criticality}.
\end{proof}
In a conflict digraph $\Dc$, if an arc is critical, then it must belong to an induced dicycle. Otherwise, it can be removed without affecting the capacity region. Thus, if a conflict digraph is critical, then it must be strongly connected.
\begin{example}
\normalfont
In Fig. \ref{fig:strong}(b), the arcs lying inside the shadow, which belong to the unique and shortest chordless dicycle, are critical. In Fig. \ref{fig:thm1-ex1}(b), the arcs forming the clique, which is the largest and unique clique, are critical.
\end{example}

\section{Linear Optimality}
\label{sec:linear-opt}
In view of the equivalence between the TIM-MP and SIC problems with linear coding schemes, in this section, we restrict ourselves to linear coding schemes, and consider the optimality of orthogonal access for some instances in terms of linear symmetric DoF $d_{\sym,l}$ for TIM-MP (or linear symmetric rate $R_{\sym,l}=\frac{1}{\beta_l^{\rm SIC}}$ for SIC). 

\subsection{Linear Optimality of Some MNI Matrices}
In what follows, we show that, for two network topologies that are not included in Theorem \ref{theorem:OA-optimal}, orthogonal access is linearly optimal for both TIM-MP and SIC problems.

\begin{theorem} [Linear Optimality for Some Special Structures]
\label{theorem:linear-opt}
For the TIM-MP and SIC instances with dicycle-vertex incidence matrix $\Cm_5^2$ and $\Jm_3$, orthogonal access achieve the optimal linear symmetric DoF/rate, where 
\begin{gather}
d_{\sym,l}(\Dc(\Cm_5^2)) =R_{\sym,l}(\Dc(\Cm_5^2))= \frac{2}{5},\\
d_{\sym,l}(\Dc(\Jm_3))=R_{\sym,l}(\Dc(\Jm_3)) = \frac{2}{5}.
\end{gather}
\end{theorem}
\begin{proof}
See Appendix \ref{proof:linear-opt}.
\end{proof}

The conflict digraphs with dicycle-vertex incidence matrices $\Cm_5^2$ and $\Jm_3$ are shown in Fig.~\ref{fig:nonideal}(a) and in Fig.~\ref{fig:nonideal}(b), respectively.
\begin{figure}[htb]
\centering
\includegraphics[width=0.8\columnwidth]{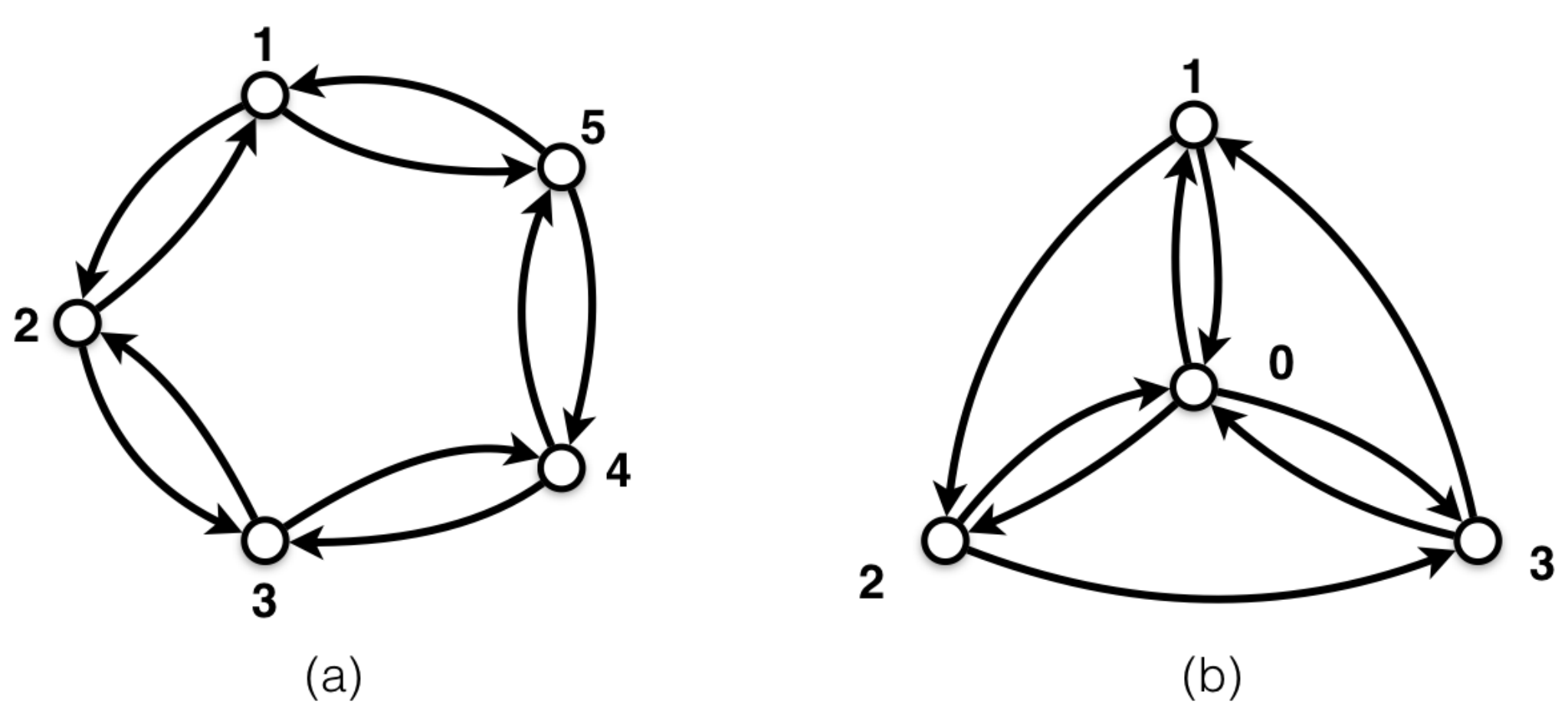}
\caption{Two types of conflict digraphs $\Dc$ with non-ideal dicycle-vertex incidence matrices, e.g., (a) circulant matrix $\Cm_5^2$ and (b) degenerate projective plane $\Jm_3$.}
\label{fig:nonideal}
\end{figure}

\subsection{Small Networks with Reduction}
Together with reducibility and criticality of the TIM-MP/SIC problems, we show the linear optimality of orthogonal access for all 4-user network topologies (in total 218 non-isomorphic conflict graphs).
\begin{theorem}[Linear Optimality for Small Networks]
\label{theorem:4-user}
For TIM-MP/SIC problems up to 4 users, orthogonal access achieves the linearly optimal symmetric DoF/rate for all topologies.
\end{theorem}
\begin{proof}
See Appendix \ref{proof:4-user}.
\end{proof}

Thanks to the reducibility and criticality, the network topologies that can be reduced to 3-user case are already done in Corollary \ref{cor:3-user}, and we only need to focus on the case when each vertex is irreducible and each arc is critical. This substantially reduces the number of non-isomorphic instances that need to be considered from 218 to 6.
It can also be checked that, for the TIM-MP/SIC instances up to 5 users, orthogonal access achieves the optimal linear symmetric DoF/rate for {\em almost} all topologies, i.e., except two out of in total 9608 non-isomorphic ones. Those two instances are $\Dc(\Jm_4)$ and $\Dc(\Cm_5^3)$ shown in Fig. \ref{fig:mni}(a) and Fig. \ref{fig:mni}(b), respectively.
Such a reduction approach based on reducibility and criticality has great potential to further identify the symmetric DoF/rate for larger networks, although the number of non-isomorphic topologies increases dramatically as the number of users increases.

\section{Discussion: Message Passing Overhead and Achievable Rate Tradeoff}
\label{sec:tradeoff}
From the previous section, we have seen that decoded message passing is so powerful that it leads to orthogonal access with almost optimal DoF, if no constraints are imposed on message passing. In practice, however, message passing may incur some cost. For example, in the uplink of a cellular system there may be limitations of the usage of the wired network connecting the base station receivers. Then, it is meaningful to study the case where a limited number of messages can be passed along from each receiver. In this section, we first consider the case when only one passing message in the entire network can help improve DoF region of TIM problem, followed by the generalization to arbitrary number of passing messages and the formulation of a matrix completion problem for the tradeoff between achievable symmetric DoF and the passing messages overhead. 
\subsection{When Does One Message Passing Help?}
\label{sec:helpful}
Given a specific decoding order, the decoded message passing is also determined. After interference cancelation with passed messages, the TIM-MP problem becomes a modified TIM problem with some interfering links removed. According to the equivalence between TIM and index coding problems \cite{Jafar:2013TIM}, a decoding order corresponds to the arc removal in the conflict digraph $\Dc$ of the TIM problem, and equivalently the arc adding in the side information digraphs $\bar{\Dc}$ of the index coding problem. 
A natural question is whether message passing helps in the sense that the corresponding arc removal from $\Dc$ increases the DoF region of the TIM problem. The following theorem offers a sufficient condition to this question. 
\begin{theorem}
\label{theorem:helpful}
A message passing is helpful, if the addition of the corresponding arc in the side information digraph $\bar{\Dc}$ forms new dicycles as induced sub-digraphs.
\end{theorem}

\begin{proof}
See Appendix \ref{proof:helpful}.
\end{proof}

\begin{remark}
\normalfont
The newly formed dicycle is not necessarily unique. It may form multiple dicycles. As message passing will not be harmful, as long as a new dicycle is formed, the DoF region will be enlarged.
\end{remark}

While the above condition is only sufficient in general, it is also necessary for chordal bipartite networks. We have the following corollary, whose proof is presented in Appendix \ref{proof:chordal-helpful}.

\begin{corollary}
\label{cor:chordal-helpful}
For chordal bipartite networks, a message passing is helpful if and only if the addition of the corresponding arc in the side information digraph $\bar{\Dc}$ forms new dicycles as induced sub-digraphs.
\end{corollary}

\begin{example}
\normalfont
Let us consider a simple network topology shown in Fig. \ref{fig:critical}(a), which is a chordal bipartite network \cite{Yi:Fractional}. Its conflict digraph $\Dc$ and the side information digraph $\bar{\Dc}$ of the corresponding index coding problem are shown respectively in (b) and (c). The DoF region for this network topology is $d_1 + d_2 + d_3 \le 1$ and the symmetric DoF are $\frac{1}{3}$.
In Fig. \ref{fig:critical}(d), the addition of the arc $(1,3)$ in $\bar{\Dc}$ forms a new dicycle $C_3$, which increases the symmetric DoF to $d_{\sym}=\frac{1}{2}$. By removing the uni-directed arc $(1,2)$ in ${\Dc}$ in Fig. \ref{fig:critical}(e), the DoF region is enhanced to $\{d_1+d_3 \le 1, \; d_2+d_3 \le 1\}$. The DoF region of Fig. \ref{fig:critical}(f) after removing the arc $(3,1)$ in $\Dc$ is still $d_1 + d_2 + d_3 \le 1$, because it is not uni-directed in $\Dc$ nor forming a new dicycle by its addition in $\bar{\Dc}$. \hfill $\lozenge$ 
\begin{figure}[htb]
 \centering
\includegraphics[width=0.8\columnwidth]{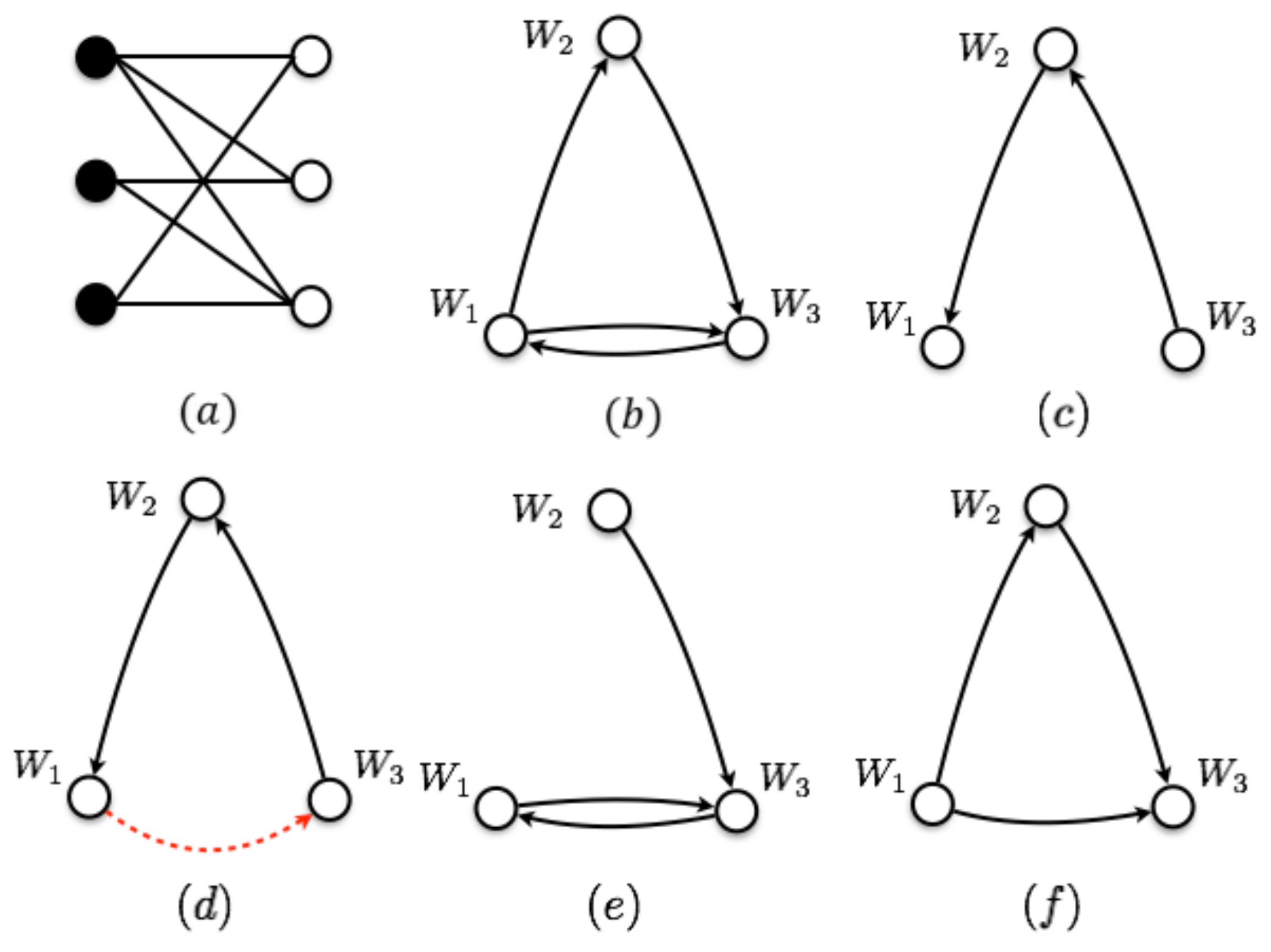}
\caption{(a) A simple network topology, (b) the conflict digraph $\Dc$, (c) the side information digraph $\bar{\Dc}$, (d) the addition of arc $(1,3)$ that forms a dicycle in $\bar{\Dc}$, (e) the removal of the uni-directed arc $(1,2)$ in $\Dc$ which also forms a new dicycle $(1,2)$ in $\bar{\Dc}$, and (f) the removal of the arc $(3,1)$ in $\Dc$ is equivalent to adding the arc $(3,1)$ in $\bar{\Dc}$.}
\label{fig:critical}
\end{figure}
\end{example}

As a side remark, a similar setting was investigated for the index coding problem in \cite{Critical-JSAC,Kim-Critical} under the name of ``critical index coding''. The goal of critical index coding problem is to figure out if one arc in the side information digraph is critical in the sense that its removal reduces the capacity region of the index coding problem. It can be regarded as the dual problem of ours in terms of arc removing/adding.

\subsection{The Achievable Rate with Message Passing Constraint}
Let $p$ be the total number of message passing budget.
The question to ask is, given $p\ge 2$, how to choose the most efficient $p$ passing such that the achievable DoF is maximized? It is a generalization of the previous subsection, in which $p=1$. Given a conflict digraph, it is to choose $p$ arcs in conflict digraphs $\Dc$, such that after removing such $p$ arcs, the broadcast rate of the index coding problem of the resulting conflict digraph is improved.

As shown in \cite{hassibi2014topological}, the TIM problem can be formulated as a matrix completion problem, by minimizing the rank of the binary matrix $\Am$ that fits the conflict digraph $\Dc$, where 
\begin{align}
\Am_{ij} = \left\{ \Pmatrix{1, & \text{if } i=j,\\0,& \text{if $(i,j) \in \Ac(\Dc)$} \\ *, & \text{otherwise}} \right.
\end{align}
with $*$ being an indeterminate value. The solution to the rank minimization problem gives a realization of $\Am$ with all entries determined. By matrix decomposition $\Am=\Um \Vm$, we have $\Um \in \CC^{K \times r}$ and $\Vm \in \CC^{r \times K}$, where $K,r$ are respectively the number of users and the minimal rank of $\Am$. Then, we can assign the columns of $\Vm$ to the transmitters and the rows of $\Um$ to the receivers as the precoding vectors during $r$ symbol extension (i.e., time slots). This gives a feasible coding scheme for the corresponding TIM instance, by which the symmetric DoF $\frac{1}{r}$ is achievable. The details can be found in \cite{hassibi2014topological}.

Similarly, we can also formulate the TIM-MP problem as a modified matrix completion problem.
A message passing from $i$ to $j$ induces a change of $\Am_{ij}$ from zero to an indeterminate value. Hence, the tradeoff between achievable symmetric DoF and message passing overhead is to minimize the rank of $\Am$ up to at most $p$ changes of the zero elements.
We introduce a message passing indicator matrix $\Xm$, such that
\begin{align}
\Xm_{ij} = \left\{ \Pmatrix{1, & \text{if $\exists$ message passing from $i$ to $j$},\\0,& \text{otherwise.}} \right.
\end{align}
We also define a matrix $\Bm$, such that
\begin{align}
\Bm_{ij} = \left\{ \Pmatrix{1, & i=j,\\ 0,& \text{if $\Am_{ij}=0$ and $\Xm_{ij}=0$},\\ *, & \text{otherwise.}} \right.
\end{align}

For a given message passing overhead $p$, the achievable symmetric DoF $\frac{1}{r}$ can be obtained by solving the following optimization problem
\begin{subequations}
\begin{align}
r=\min_{\Xm} \rank & \quad \Bm \\
\st~ & \quad \norm{\Xm}_0 \le p \label{mat-comp-overhead}\\
&\quad \{(i,j): \Xm_{ij}=1 \} \text{~is} \nonumber \\& \qquad \qquad \text{an acyclic set in $\Dc$} \label{mat-comp-passing}
\end{align}
\end{subequations}
where \eqref{mat-comp-overhead} limits the number of message passing within $p$, and \eqref{mat-comp-passing} makes sure the message passing is feasible. It is impossible to pass the message that should be decoded later to a receiver that need to decode earlier, so that the sub-digraph induced by the passed messages should be an acyclic set.

{In particular, given a sufficiently large budget of message passing $p$, the solution to this optimization problem yields an interference alignment solution to the TIM-MP problem.
This solution is not necessarily a one-to-one alignment, and therefore it may improve over orthogonal access. On the other hand, when the budget of message passing $p$ is constrained, the orthogonal access solution under the TIM-MP setting may not be feasible, and therefore one-to-one alignment becomes useful again.}

The optimization problem is non-convex with a combinatorial nature, and thus hard to solve. Some existing algorithmic methods for matrix completion (see \cite{hassibi2014topological} and references therein) can be applied here to obtain the approximate solutions. The algorithm design and convergence analysis are interesting problems yet beyond the scope of this paper. Instead, we present a typical example in the following for illustration.
\begin{example}
\normalfont
Let us consider a 4-user triangular network as shown in Fig. \ref{fig:tradeoff}(a). The tradeoff between the reciprocal symmetric DoF $r=\frac{1}{d_{\sym}}$ and the number of passing messages $p$ is also illustrated in Fig. \ref{fig:tradeoff}(b). When $p=0$, it is a conventional TIM problem without message passing, and thus $r=4$. When $p=1$ and $p=2$, the removal of any one and two cross links enables a vertex coloring with 3 and 2 colors, thus yielding $r=3$ and $r=2$, respectively. To completely remove all interfering links, we need 6 passing messages.  \hfill $\lozenge$
\begin{figure}[htb]
 \centering
\includegraphics[width=0.8\columnwidth]{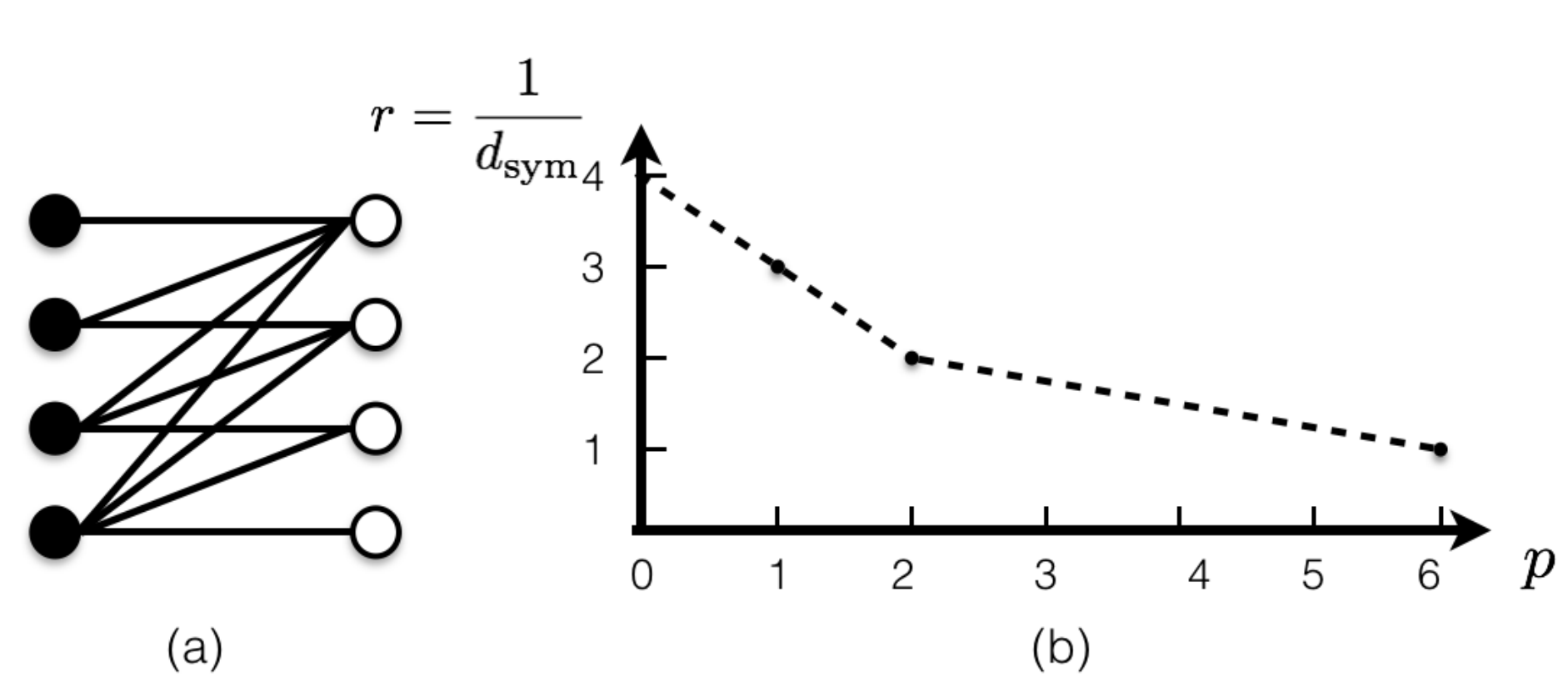}
\caption{(a) The network topology of 4-user triangular network, (b) the tradeoff between the reciprocal symmetric DoF $r=\frac{1}{d_{\sym}}$ and the number of passing messages $p$.}
\label{fig:tradeoff}
\end{figure}
\end{example}

\section{Conclusion}
The topological interference management with decoded message passing (TIM-MP) problem in partially-connected uplink cellular networks has been considered, where transmitters have only access to topological information without knowing the channel realizations and receivers are able to pass their decoded messages once they are decoded to other receivers. By modeling the interference pattern as a conflict directed graph, we have bridged the orthogonal access in this setting to the acyclic set coloring on directed conflict graphs. With the aid of polyhedral combinatorics, we have shown that orthogonal access achieves the optimal DoF region for certain classes of networks. The relation to index coding has been also investigated by connecting TIM-MP to a generalized index coding problem where a successive decoding and message passing policy at receivers is allowed. Reducibility and criticality were also studied to reduce large-size problems to smaller ones, by which the linear optimality of orthogonal access is shown for small-size networks up to four-user with all possible topologies. The practical issue on the tradeoff between achievable DoF and message passing was also discussed, from the usefulness of one message passing to the general case formulated by a matrix completion problem.

Yet, fundamental limits of decoded message passing in the TIM-MP setting are not fully understood. With message passing, whether or nor orthogonal access is sufficient to achieve the (linearly) optimal DoF region for all network topologies is still an intriguing yet challenging problem.
Although one-to-one interference alignment does not benefit beyond orthogonal access in this TIM-MP setting, it is also interesting to see if subspace alignment has gains.
The tradeoff between the network performance (e.g., achievable DoF) and the overhead of backhaul links (e.g., the number of passed messages) is also a research avenue of great  interest. %

\appendix
\subsection{Graph Theory}
\label{appendix:graph}
In what follows, the definitions and results pertaining to graph theory are briefly recalled for readers' convenience. Interested readers are suggested referring to the text books \cite{combinopt,GT} and papers \cite{SPDT,Dichrom} for more details.

In this paper, we mainly focus on the directed graphs (digraphs), usually denoted by $\Dc=(\Vc,\Ac)$ with a vertex set $\Vc$ and an arc set $\Ac$ consists of ordered pairs of vertices. An arc $(u,v) \in \Ac$ with $u,v \in \Vc$ is a directed edge from $u$ to $v$.
The underlying undirected graph $\Gc$ of $\Dc$ is created in such a way that any two vertices are joint with an edge in $\Gc$ if and only if there exists at least one arc between them in $\Dc$. We denote by $S(\Dc)$ the symmetric part of $\Dc$, which is an undirected graph such that any two vertices $u,v$ are joint with an edge, if and only if both $(u,v) \in \Ac$ and $(v,u)\in \Ac$ hold. An arc either $(u,v) \in \Ac$ or $(v,u) \in \Ac$ but not both is referred to as a uni-directed arc, otherwise it is bi-directed. The complement of a digraph $\Dc=(\Vc,\Ac)$, denoted by $\bar{\Dc}=(\Vc,\bar{\Ac})$, has the same vertex set $\Vc$ and $(u,v) \in \bar{\Ac}$ if and only if $(u,v) \notin \Ac$.
A sub-digraph of $\Dc$ induced by vertex set $\Uc$, denoted as induced sub-digraph $\Dc[\Uc]$, is such that, $\forall u,v \in \Uc$, an arc $(u,v) \in \Ac(\Dc[\Uc])$ if $(u,v) \in \Ac(\Dc)$.
The in-degree [resp. out-degree] of the vertex $v$, denoted by $d^-({\Dc},v)$ [resp. $d^+({\Dc},v)$], is the number of vertices $u \in \Vc$ such that $(u,v) \in \Ac$ [resp. $(v,u) \in \Ac$]. The maximum in-degree, denoted by $\Delta^-(\Dc)$, is the maximum value of $d^-({\Dc},v)$ over all vertices $v$.

A directed cycle (dicycle) with length $n$, denoted by $C_n=(v_0,v_1,\dots,v_{n-1})$,  refers to the induced sub-digraph with arcs $\{(v_i,v_{(i+1) \mod n}), i \in \{0,\dots,n-1\} \}$, beyond which there do not exist any other arcs. $C_n$ is an odd cycle if $n$ is odd, and an even cycle if $n$ is even.
A digraph (sub-digraph) is acyclic if it does not contain any dicycles. A directed acyclic sub-digraph is referred to as an acyclic set. 
Every directed acyclic set has at least a vertex of in-degree 0 and at least a vertex of out-degree 0. Every directed acyclic set has an acyclic ordering of its vertices, i.e., there exists an ordering of vertices $v_1,v_2,\dots, v_n$ in the acyclic set $\Dc$, such that for every arc $(v_i,v_j) \in \Ac$, we have $i<j$.

A directed path $v_0 \to v_n$ is a set of arcs $\{(v_0,v_{1}), \dots, (v_i,v_{i+1}),\dots,(v_{n-1},v_{n})\}$ connecting $v_0$ to $v_n$.
A digraph $\Dc = (\Vc, \Ac)$ is strongly-connected (or strong) if for every two distinct vertices $v_i, v_j \in \Vc$, there exist directed paths $v_i \to v_j$ and $v_j \to v_i$ in $\Dc$. The directed cycles and cliques are strongly-connected. A strong component of $\Dc$ is a maximal induced sub-digraph which is strongly-connected.
A partition $\Vc=\{\Vc_1,\dots,\Vc_s\}$ with $\Vc=\cup_{i=1}^s \Vc_i$ and $\Vc_i \cap \Vc_j=\emptyset$ for all $i \ne j$ is called a strong component decomposition, if every sub-digraph $\Dc[\Vc_i]$ is a strong component.

Formally, the dichromatic number \cite{Dichrom} of a digraph $\Dc$, $\chi_A(\Dc)$, is the smallest cardinality $\abs{\Xc}$ of a color set $\Xc$, so that it is possible to assign a color from $\Xc$ to each vertex of $\Dc$ such that for every color $c \in \Xc$ the sub-digraph induced by the vertices colored with $c$ is acyclic.  Dichromatic number $\chi_A$ generalizes the notion of chromatic number $\chi$ from graphs to digraphs. The subgraph induced by the vertices with the same color in graphs is an independent set, while it is an acyclic set in digraphs. $\chi_A(\Dc) \le \chi(\Gc)$ holds for any digraph $\Dc$ and its underlying undirected graph $\Gc$, because independent sets in $\Gc$ are special acyclic sets, such that any proper coloring of $\Gc$ is a proper acyclic set coloring of $\Dc$ \cite{SPDT}.

A clique of a digraph is a sub-digraph in which for any two distinct vertices $u$ and $v$ both arcs $(u,v)$ and $(v,u)$ exist. 
The {\em maximal} clique of a (di)graph is a clique that cannot be contained by another clique with larger size as an induced sub-(di)graph. We denote by $Q_n$ a clique with $n$ vertices.
A clique in the digraph $\Dc$ is also a clique in the symmetric undirected graph $S(\Dc)$. Each node is a clique with size 1. The clique number $\omega(\Dc)$ is the size of the largest clique in $\Dc$. Obviously, $\chi_A(\Dc) \ge \omega(\Dc)$, because any two vertices in the clique cannot be in the same acyclic set.  

A chordless cycle is a set of vertices and edges that form a closed loop without chord. A chord is an edge that connects two non-adjacent vertices of a cycle. The length of a cycle is the number of vertices in this cycle. 
The hole is the chordless cycles with length greater than 4, and the antihole is the complement of hole.
The odd hole is a hole with odd length, and the odd antihole is its complement.
An {\em undirected graph} is perfect \cite{combinopt} if and only if it does not contain odd holes nor odd antiholes as induced subgraphs.
A chordal bipartite network is a bipartite graph network that does not contain induced chordless cycles of length 6 or more.

A {\em digraph} $\Dc$ is perfect if, for any induced sub-digraph $\Dc[\Uc]$ of $\Dc$ induced by any vertex subset $\Uc$, $\chi_A(\Dc[\Uc])=\omega(\Dc[\Uc])$. By the Strong Perfect Digraph Theorem in \cite{SPDT}, a digraph $\Dc=(\Vc,\Ac)$ is perfect if and only if the undirected symmetric part $S(\Dc)$ is perfect and $\Dc$ does not contain any dicycle $C_n$ with $n \ge 3$ as induced sub-digraph. Specifically, a {\em digraph} is perfect if and only if it does neither contain a filled odd hole, nor a filled odd antihole, nor a dicycle $C_n$ with $n \ge 3$ as induced sub-digraph \cite{SPDT}. %
The filled odd hole is a digraph with symmetric part being an odd hole, and the filed odd antihole is the complement of the filled odd hole.

\subsection{Polyhedral Combinatorics}
\label{appendix:polyhedral}
A (convex) polyhedron in $\RR^n$ can be defined as the solution set of a finite system of linear inequalities with $n$ variables.
A polytope is a bounded polyhedron. The $\Ic$-$\Jc$ incidence matrix $\Am$ is a matrix that shows the relationship between two sets of objects $\Ic$ and $\Jc$, such that 
\begin{align}
\Am_{ij} = \left\{ \Pmatrix{1, & \text{if $i \in \Ic$ has a relation to $j \in \Jc$,} \\ 0, & \text{otherwise.}} \right.
\end{align}

A polytope is integral if all its extreme points have only integer-valued coordinates. %
The set packing/covering polytopes are the most important polytopes in polyhedral combinatorics.
The set packing polytope is given by
\begin{align} 
\mathscr{P}(\xv)&=\left\{\xv \in [0,1]^n: \Am \xv \le \one \right\},
\end{align}
and the set covering polytope is given by
\begin{align} 
\mathscr{Q}(\xv)&=\left\{\xv \in [0,1]^n: \Bm \xv \ge \one \right\}.
\end{align}

A row vector $\rv$ of a matrix is said to be dominating if there exists another different row $\sv$ such that $\rv \ge \sv$ for set covering polytopes and $\rv \le \sv$ for set packing polytopes. In other words, the linear inequality constraint associated with a dominating row is redundant and can be implied by others. %

A submatrix of $\Am$ is a minor of $\Am$ if it can be obtained from $\Am$ by successively deleting a column $j$ and the rows with a `1' in column $j$. Given a set $\Ac$, the minor of a matrix $\Am$ is the submatrix of $\Am$ that results from removing all columns indexed in $\Ac$ and all the dominating rows that may occur.

For the convenience of the study in integrality of set packing and covering polytopes, perfect and ideal matrices were introduced.
The definitions and characterizations of perfect, ideal and balanced matrices are summarized as follows \cite{conforti2001perfect}.

A matrix is perfect if the set packing polytope is integral.
Let $\Am$ be a binary matrix and $\Gc$ be an undirected graph. Let the columns of $\Am$ correspond to the vertices of $\Gc$ and let the rows of $\Am$ be the incidence vectors of the maximal cliques of $\Gc$. Then, $\Am$ is a perfect matrix if and only if $\Gc$ is a perfect graph \cite{stabqstab}. %
The set packing polytope is integral if and only if $\Am$ is the {maximal} clique-node incidence matrix of a perfect graph.

A matrix is ideal if the set covering polytope is integral \cite{cornuejols1994ideal}. Ideal matrices are also known as width-length matrices \cite{lehman1979width,lehman1990width}, or matrices with (weak) max-flow min-cut property \cite{seymour1977matroids}. If a matrix is ideal, so are all its minors.
It is an open problem to fully characterize all families of ideal matrices. An alternative way is to consider the ``smallest'' possible matrices that are not ideal, referred to as minimally nonideal (MNI) matrices \cite{lehman1979width,lehman1990width}. A matrix is MNI, if it is not ideal but all its proper minors are.
In other words, a matrix is MNI if (1) it does not contain a dominating row, (2) it is nonideal, and (3) the coordinates of any extreme point of the set covering polytope are either all integral or all fractional (but not both).  
Thus, an alternative characterization of ideal matrix is that, $\Am$ is ideal if and only if $\Am$ does not contain a MNI minor \cite{padberg1993lehman}. 

There are some known MNI matrices, although the complete classification is an open problem.
The circulant matrix $\Cm_n^r$ is a $n \times n$ binary matrix with column indexed by $\{1,2,\dots,n\}$ and rows equal to the incidence vectors of $\{j, j+1, \dots, j+r-1\} \mod n$, i.e., $\Cm_{ij}=1$ if $j \in \{i,i+1,\dots,i+r-1\} \mod n$ and $\Cm_{ij}=0$ otherwise. The degenerate projective plane $\Jm_n$ for $n \ge 2$ is a square $(n+1) \times  (n+1)$ binary matrix with columns indexed by $\{0,1,\dots,n\}$ and rows equal to the incidence vectors of $\{1,\dots,n\}$, $\{0,1\}$, $\{0,2\}$, $\dots$, $\{0,n\}$. It has been shown that $\Cm_n^2$ for $n \ge 3$ odd and $\Jm_n$ for $n \ge 2$ are MNI matrices \cite{lehman1979width,lehman1990width}.
Cornuejols and Novick \cite{cornuejols1994ideal} proved that there are exactly 10 MNI circulant matrices with $k$ consecutive 1's $(k \ge 3)$. The Fano matrix $\Fm_7 \in \{0,1\}^{7 \times 7}$ is a circulant matrix with the initial row vector $(1,1,0,1,0,0,0)$. 
Lehman \cite{lehman1990width} gave the properties of MNI matrices by proving that their fractional set covering polyhedron has a unique fractional extreme point.

A binary matrix is balanced if, and only if, for each submatrix both the set covering polytope and the set packing polytope are integral \cite{berge1972balanced}. A matrix is balanced if and only if it and all its submatrices are perfect or, equivalently, if and only if it and all its submatrices are ideal. %
An odd hole in a binary matrix is a square submatrix of odd order with two ones per row and per column.
A binary matrix is balanced if it does not contain odd hole as a submatrix.
The balanced matrices can be recognized in polynomial time \cite{zambelli2005polynomial}.

A matrix is totally unimodular (TU) if and only if every square submatrix has determinant equal to 0, +1, -1. A matrix $\Am$ is TU if and only if the polyhedron $\{\xv \ge 0: \Am \xv \le \bv\}$ is integral for every integral vector $\bv$. A TU matrix is both perfect and ideal.
If each row of a 0/1 matrix (up to permutation) has consecutive 1's, then this matrix is TU.  Seymour proved a full characterization of all TU matrices in \cite{seymour1980decomposition},  that is, a matrix is TU if and only if it is a certain natural combination of some network matrices and some copies of a particular 5-by-5 TU matrix.

\subsection{Index Coding}
\label{appendix:index}
The index coding problem considers the transmission in a noiseless broadcast channel, where each receiver wants one message from the transmitter and holds some other receivers' desired messages as side information. The goal is to find out the minimum number of transmissions such that all receivers are able to decode their own messages {\em synchronously}.  A formal definition is as follows.

A $(t,r)$ index code $\mathbb{C}$ with side information index sets $\{\Sc_1,\dots,\Sc_n\}$ is defined as follows:
\begin{itemize}
\item an encoding function, $\phi: \{0,1\}^{tn} \mapsto \{0,1\}^r$ at the transmitter that encodes $n$-tuple of messages $x^n$ to a length-$r$ index code.
\item a decoding function at each receiver $j$, $\psi_j: \{0,1\}^r \times \{0,1\}^{t\abs{\Sc_j}} \mapsto \{0,1\}^{t}, \; \forall j$ that decodes the received index code back to $x_j$ with side information $x(\Sc_j)$ held at receiver $j$.
\end{itemize}

The side information index sets are usually represented by a side information digraph $\bar{\Dc}$, whose complement is the conflict digraph $\Dc$.
\begin{definition} [Side Information Digraph]
For the index coding problem with message set $\Wc=\{W_1,\dots,W_n\}$ and side information index sets $\{\Sc_j, \forall j\}$, the side information digraph $\bar{\Dc}=(\Vc,\bar{\Ac})$ is such that $\Vc=\Wc$ and $(i,j) \in \bar{\Ac}(\bar{\Dc})$ if and only if $W_i \in \Sc_j$. 
\end{definition} 

Thus, a rate $\beta'({\Dc},\mathbb{C})=\frac{r}{t}$ is said to be achievable if there exists an index code $(t,r)$, such that
\begin{align}
\psi(\phi(x^n),x(\Sc_j)) = x_j, \quad \forall j
\end{align}
The broadcast rate of the index coding problem is defined as
\begin{align}
\beta'({\Dc}) = \inf_t \inf_{\mathbb{C}} \beta'({\Dc},\mathbb{C}).
\end{align}  

We introduce an outer bound of the index coding problem, which will be frequently used in the proofs.
\begin{lemma} (Maximal Acyclic Induced Subgraph (MAIS) Outer Bound \cite{Index2011})
For the index coding problem with side information digraph $\bar{\Dc}$, the broadcast rate is upper bounded by
\begin{align}
\beta \le \alpha_A(\bar{\Dc})
\end{align}
where $\alpha_A(\bar{\Dc})$ the largest size of all acyclic sets in $\bar{\Dc}$.
\end{lemma}

\begin{definition}[Partial Clique]
\label{def:pc}
A conflict sub-digraph ${\Dc}[\Qc]$ is a $k$-partial clique if and only if $\forall v_q \in \Qc$, $d^-({\Dc}[\Qc],v_q) \le k$, and $\exists v_q^* \in \Qc$, $d^-({\Dc}[\Qc],v_q^*)=k$.
\end{definition}
\begin{lemma}[Partial Clique Covering \cite{ISCOD}]
The broadcast rate of the index coding problem with conflict digraph ${\Dc}=(\Vc,{\Ac})$ is upper bounded by
\begin{align}
\beta \le \min_{\{\Vc_1, \dots, \Vc_s\}} \sum_{i=1}^s (k_i +1),
\end{align}
with partial clique covering, where the minimum is over all partitions of $\Vc=\{\Vc_1, \dots, \Vc_s\}$, and for all $i = 1,\dots,s$, ${\Dc}[\Vc_i]$ is a $k_i$-partial clique with $k_i=\Delta^-(\Dc[\Vc_i])$.
\end{lemma}

\subsection{Proof of Corollary~\ref{cor:1dof}}
\label{proof:cor-1dof}
The achievability is due to acyclic set coloring. For the ``if'' part, if the conflict digraph $\Dc$ is acyclic, we have $\chi_A(\Dc)=1$. In addition, $d_{\sym} \le 1$ should also hold, because if we allow full receiver cooperation, the enhanced channel is a $K$-user multiple access channel, such that $d_{\sym} \le 1$ even if full CSI is available at transmitters. For the ``only if'' part, we prove it by contraposition. If there exists a cycle in the conflict digraph, no message passing policy can remove the cycle completely, because for a message that will be decoded later cannot provide its decoded message before its decoding. As such, interference still exists in the network after message passing, and thus $d_{\sym}<1$ for sure. By contraposition, if $d_{\sym}=1$, the conflict digraphs should be acyclic. This completes the proof for both ``if'' and ``only if'' parts.

\subsection{Proof of Corollary~\ref{cor:half-dof}}
\label{proof:cor-half-dof}
The achievability is due to acyclic set coloring. It was recently proven in \cite[Theorem 3]{dichrom-cycles} that, if there exist two integers $k$ ($k \ge 2$) and $r$ ($1 \le r \le k$) such that a digraph $\Dc$ contains no dicycle of length $r\mod k$, then $\chi_A(\Dc) \le k$. When $k=2$, we have $r=1,2$. So, the condition turns out to be that, if $\Dc$ either contains no directed odd cycle or no directed even cycle, $\chi_A(\Dc) \le 2$. As single-round message passing is considered, $\chi_A(\Dc)$ should be an integer. 
Because $\chi_A(\Dc)=1$ if and only if $\Dc$ is a directed acyclic graph, it follows that $\chi_A(\Dc)=2$ if $\Dc$ contains only directed odd cycles or even cycles. By Theorem \ref{theorem:OA}, $d_{\sym} \ge \frac{1}{2}$.  For the converse, because the conflict digraphs contain cycles no matter whether odd or even, they are not directed acyclic graphs, so that $d_{\sym}<1$ according to Corollary \ref{cor:1dof}. Because the TIM-MP problem with single-round message passing results in an enhanced TIM problem, it follows that $d_{\sym} \le \frac{1}{2}$. As such, we have $d_{\sym}=\frac{1}{2}$. This completes the proof.

\subsection{Proof of Corollary \ref{cor:regular}}
\label{proof:cor-regular}
The achievability is due to fractional acyclic set coloring of the conflict digraphs.
For the $(K,L)$ regular network, the arcs $(j,i)$ with $i=j+1, j+2, \dots, j+L-1$ belong to the conflict digraph $\Dc$.
It can be easily checked that the sub-digraph of the conflict digraph induced by any $K-L+1$ neighboring vertices is an acyclic set. As such, we are able to color the acyclic sets with in total $K$ colors and each vertices can be assigned with $K-L+1$ distinct colors. On average, the number of colors allocated to each vertex (i.e., message) is $ \frac{K}{K-L+1}$. Thus, $\chi_{A,f} \le \frac{K}{K-L+1}$ and in turn $d_{\sym} \ge \frac{K-L+1}{K}$. 

\subsection{Proof of Lemma \ref{lemma:IA-nohelp}}
\label{proof:IA-nohelp}
For a digraph $\Dc=(\Vc,\Ac)$ and its optimal acyclic set partition $\Dc=\{A_1,\dots,A_s\}$, we construct another digraph $\Dc'=(\Vc',\Ac')$ with vertices $v'$ representing the acyclic set $A_v$ of $\Dc$, and two vertices $u'$ and $v'$ are connected with an arc $(u',v')$ if there exists an arc in $\Dc$ from a vertex in $A_u$ to a vertex in $A_v$, where $A_u \ne A_v$. It is clear that $\chi_{A,f}(\Dc')=\chi_{A,f}(\Dc)$, because the numbers of required colors in total for $\Dc'$ and $\Dc$ are the same according to the construction of $\Dc'$. According to Definition \ref{def:loc-dichrom}, we conclude that $\chi_{LA,f}(\Dc')=\chi_{LA,f}(\Dc)$. It is because, given the optimal acyclic set partition, the number of colors in the most colorful neighborhood of a vertex equals to the number of neighboring acyclic sets of the one which this vertex belongs to. Otherwise, this vertex should join into other acyclic sets, which contradicts with the optimal acyclic set partition.
It is clear that for every arc $(u',v') \in \Ac'$, there always exists an arc $(v',u') \in \Ac'$, because otherwise the union of $A_u$ and $A_v$ is still acyclic and can be united into one acyclic set which contradicts with the assumption that $\Dc=\{A_1,\dots,A_s\}$ is the optimal acyclic set partition. For the digraph $\Dc'$ with all arcs being bi-directed, we have $\chi_{A,f}(\Dc')=\chi_f(U(\Dc'))$ and $\chi_{L,f}(U(\Dc'))=\chi_{LA,f}(\Dc')$, where $U(\Dc')$ is the underlying undirected graph of $\Dc'$ \cite{CircularNo}. In addition, for the undirected graph $U(\Dc')$, we have $\chi_f(U(\Dc'))=\chi_{L,f}(U(\Dc'))$ \cite[Theorem 5]{korner2005local}. As such, we have $\chi_{A,f}(\Dc')=\chi_{LA,f}(\Dc')$ \cite{simonyi2015relations}. It follows immediately that $\chi_{A,f}(\Dc)=\chi_{LA,f}(\Dc)$, which completes the proof.

\subsection{Proof of Theorem \ref{theorem:outer-bounds}}
\label{proof:outer-bounds}
Let us first focus on the clique inequalities.
A clique $Q$ in conflict digraph $\Dc$ corresponds to a fully connected interference channel in network topology. Without CSI at the transmitters, the sum DoF are bounded above by 1, even if the message passing is allowed.
Thus, we conclude that the achievable DoF tuple should satisfy the following inequalities:
\begin{align}
\sum_{i \in Q} d_i \le 1,
\end{align}
for each clique $Q \in \Qc$.  
Note that each vertex in $\Dc$ is also a clique with size 1. This clique inequalities also imply the individual DoF inequality $d_k \le 1$ for all $k$. The clique inequality of the maximal clique size dominates, because it implies other clique inequalities associated the sub-cliques with smaller size.

Next, for cycle inequalities, a dicycle without chord means it does not contain any smaller dicycles as induced sub-digraphs. A dicycle $C$ in the conflict digraph $\Dc$ corresponds to a cyclic Wyner interference channel. There must exist a receiver who does not obtain any passed message for decoding, such that the interference from his neighbor cannot be canceled out, and thus sum DoF of these two messages are bounded above by 1. As such, the overall sum DoF of the messages in such a dicycle cannot exceed $\abs{C}-1$, i.e., 
\begin{align}
\sum_{k \in C} d_k \le \abs{C}-1,
\end{align} 
for all dicycles. For those dicycles with chord, there must exist smaller dicycles, so that the cycle inequality due to the smaller one implies that due to the larger one, together with the individual inequality. So, for the cycle inequalities, we only count the inequalities corresponding to the minimal dicycles. 

\subsection{Proof of Theorem \ref{theorem:perfect}}
\label{proof:perfect}
The conflict digraph $\Dc$ that contains no dicycles $C_n$ with $n \ge 3$, and $\Am$ is a perfect matrix, if and only if $\Dc$ is a perfect digraph \cite{combinopt}.
Given the outer bound of DoF region in \eqref{eq:case-I-region}, we show that its extreme points are integral when the conflict digraph is perfect. Then we show that the DoF tuple of the region is achievable by acyclic set coloring (i.e., orthogonal access with single-round message passing), and thus the whole region can be achieved by time sharing, which is analogue to fractional acyclic set coloring.

If the conflict digraph $\Dc$ is a perfect digraph, then its symmetric part $S(\Dc)$ is a perfect graph \cite{SPDT}. By \cite[Chapter 65]{combinopt}, we know that the polytope defined by clique inequalities of a perfect graph has integral extreme points, meaning that each coordinate of the extreme points of the DoF region is either 0 or 1. As shown in \cite{Yi:Fractional}, the integral extreme points of the polytope defined by clique inequalities correspond to a set of messages in an independent set that do not interfere one another. It agrees with the vertex coloring of the undirected graph. It has been also shown in \cite{SPDT} that, for a perfect digraph, a feasible vertex coloring of the symmetric part of a digraph $S(\Dc)$ is also a feasible acyclic set coloring of the digraph $\Dc$. As such, each extreme point of the outer bound of DoF region can be achieved by acyclic set coloring (orthogonal access with single-round message passing), and time sharing between vertices can achieve the entire DoF region. According to the linear program relaxation in \eqref{eq:lp-acyclic}, this time sharing of acyclic set coloring is actually fractional acyclic set coloring of digraphs. This completes the proof.

\subsection{Proof of Theorem \ref{theorem:ideal}}
\label{proof:ideal}
The outer bound is given merely by the cycle inequalities, because there are no cliques in conflict digraphs. As mentioned earlier, by replacing $d_k=1-y_k$, the outer bound of the DoF region in \eqref{eq:dof-cycle} can be written as a set covering polytope
\begin{align}
\left\{(y_{\Kc}): \Pmatrix{0 \le y_k \le 1, \forall k \in \Kc \\ \sum_{k \in C} y_k \ge 1, \forall C \in \Cc} \right\}.
\end{align} 
This parameter replacement does not change the cycle-vertex incidence matrix $\Bm$, which is an ideal matrix. So, the set covering polytope given by $\{y_k\}$ is integral, and for all extreme points of the polytope, we have $y_k=0$ or $y_k=1$. In turn, the DoF tuple $(d_k=1-y_k)$ is also a binary vector. 

For the achievability, we prove that all extreme points in the outer bound of the DoF region can be achieved by acyclic set covering. For a DoF tuple corresponding to an extreme point of DoF region, we switch off the message $k$ whose coordinate $d_k=0$, while switch on the ones with $d_k=1$. Because of the cycle inequality constraint, for each dicycles in the conflict digraph, there is at least one vertex that is switched off. The cycle inequalities collect all dicycles without chord, so that they all together ensure that, for each binary DoF tuple, the active vertices in the conflict graph cannot form any dicycles, and thus the sub-digraph induced by these active vertices is acyclic. These vertices in an acyclic set and can be assigned with the same color. Similarly to all extreme points of the DoF region, we exactly have a proper acyclic set coloring of all vertices in the conflict digraph. Time sharing among these single acyclic set coloring scheme (fractional acyclic set coloring) yields the DoF region. This completes the proof.

\subsection{Proof of Corollary \ref{corollary:theorem-4}}
\label{proof:corollary-theorem-4}
For the above two cases, the corresponding dicycle-vertex incidence matrices are balanced, as they do not contain submatrices $\Cm_n^2$ with odd $n$ where $n \ge 3$.
As a result, orthogonal access achieves the optimal DoF region according to Theorem \ref{theorem:ideal}.

\subsection{Proof of Corollary \ref{cor:3-user}}
\label{proof:cor-3-user}
The achievability is still due to fractional acyclic set coloring. The conflict graphs except the dicycle $C_3$ are all perfect digraphs. According to Theorem \ref{theorem:perfect}, orthogonal access achieves the whole capacity region. For the dicycles $C_3$, the whole DoF region $\{(d_1,d_2,d_3): 0 \le d_k \le 1, \; \forall k, d_1+d_2+d_3 \le 2\}$ is achievable, because all the vertices of the DoF region $(1,0,0),(0,1,0),(0,0,1),(1,1,0),(1,0,1),(0,1,1)$ are integers and are achievable by acyclic set coloring. In particular, the symmetric DoF tuple $(\frac{2}{3},\frac{2}{3},\frac{2}{3})$ is achievable by time sharing among $(1,1,0),(1,0,1),(0,1,1)$.

\subsection{Proof of Theorem \ref{theorem:balanced}}
\label{proof:balanced}
The integrality of the DoF region is due to \cite[Theorem 2.6]{kiraly2015extension} based on vertex covering on hypergraphs. Let us make the connection. 

Let $\Hc(\Dc)=(\Vc,\Ec)$ be a hypergraph, where the vertex set $\Vc$ corresponds to the collection of all messages or transmitter-receiver pairs, the hyperedge set $\Ec$ is the collection of dicycles $\Cc'$ in conflict digraphs $\Dc$. To avoid the redundancy of inequalities, we only count the minimal cycles without chord. For every dicycle, the cycle inequality becomes the constraint of hyperedge $E \in \Ec$, i.e., $y(E) \ge 1$, $\forall E \in \Ec$, where $y(E) \defeq \sum_{k \in E} y_k$. Let $\Gc_{\Hc}$ be the undirected graph consisting of the hyperedges of size 2 in $\Hc$. Thus, each edge in $\Gc_{\Hc}$ is the bi-directed arcs (i.e., dicycles with length 2) and $\Gc_{\Hc}$ is in fact the symmetric part of the conflict digraph $S(\Dc)$. To avoid the redundancy of inequalities, we only count the maximal cliques with size no less than 3, because size 2 clique is also a dicycle and the corresponding inequality exists already in cycle inequalities. So, for every clique, the clique inequality becomes the constraint of cliques in $\Gc_{\Hc}$, i.e., $y(Q) \ge \abs{Q}-1$, $\forall Q \in \Qc'(\Gc_{\Hc})=\Qc'(S(\Dc))$, where $\Qc'$ is the collection of all cliques $Q$ in $\Gc_{\Hc}$ with $\abs{Q} \ge 3$. 
As said before, we consider the case $\abs{C \cap Q}\le 1$ to avoid the redundancy between cycle and clique inequalities, because if $\abs{C \cap Q}\ge 2$, the constraint $y(C) \ge 1$ is redundant, which is implied by $y(Q) \ge \abs{Q}-1$.
As such, the cycle and clique inequalities can be represented with regard to the hypergraph $\Hc$ and $\Gc_{\Hc}$, given by
\begin{align}
\mathscr{P}=\left\{(y_{\Kc}): \Pmatrix{0 \le y_k \le 1, & \forall k \in \Kc \\ y(E) \ge 1, & \forall E \in \Ec \\ y(Q) \ge \abs{Q}-1, & \forall  Q \in \Qc'(\Gc_{\Hc})} \right\}
\end{align}
where $d_k=1-y_k$. 

By \cite[Theorem 2.6]{kiraly2015extension}, we know that the above polyhedron is integral if and only if $\Hc$ has no triangle-free MNI minor and $\Gc_{\Hc}$ is perfect. A minor of $\Hc$ is obtained by deletion of a node set $\Uc_1 \subseteq \Vc$ and contraction of a node set $\Uc_2 \subseteq \Vc$ where $\Uc_2$ does not contain any dicycles. The deletion operation is to remove all hyperedges incident to $\Uc_1$, that is, to remove all dicycles involving the nodes in $\Uc_1$. The contraction operation is to remove the nodes in $\Uc_2$ from all remaining hyperedges, that is, to remove the cycle inequalities involving the nodes in $\Uc_1$. A minor is called triangle-free if $\Uc_1$ covers every triangle of $\Gc_{\Hc}$. Because a triangle in an undirected graph is a clique of size 3, it follows that $\Uc_1$ covers every clique in $S(\Dc)$ with size no less than 3. A minor is MNI if the resulting hypergraph formed by the inclusionwise minimal hyperedges after deletion and contraction operations is MNI, that is, after deletion and contraction operations, the resulting hyperedge-vertex incidence matrix associated with the remaining hypergraph contains no MNI submatrices.

Reflecting to conflict digraphs $\Dc$, the condition that $\Hc$ has no triangle-free MNI minor indicates that, after removing all cliques with size no less than 3 from $\Dc$ and the associated vertices (i.e., deletion and contraction operations), the dicycle-vertex incidence matrix of the resulting induced sub-digraph of $\Dc$ contains no MNI submatrices. Together with the condition that $S(\Dc)$ is perfect, the outer bound polytope defined by clique and cycle inequalities is integral.

The achievability is still due to acyclic set coloring. The coordinates of the extreme points of the outer bound polytope $\mathscr{P}$ correspond to the on-off of the messages. If $y_k=0$, the message $k$ is on, and is off otherwise. Form the clique constraints, there are at least $\abs{Q}-1$ messages are off for the clique $Q$. From the cycle constraints, there is at least one message that is off for the dicycle $C$. The intersection between a dicycle $C$ and a clique $Q$ is at most one vertex in $\Dc$, due to $\abs{C \cap Q} \le 1$. So, for each extreme point of $\mathscr{P}$, the active vertices (i.e., with coordinate being $y_k=0$) do not contain any dicycles in $\Dc$, and thus form an acyclic set. This corresponds to an acyclic set coloring. The fractional acyclic set coloring also corresponds to time sharing among the extreme points of $\mathscr{P}$, by which the entire region of $\mathscr{P}$ is achievable. This completes the proof.

\subsection{Proof of Theorem~\ref{theorem:reducibility}}
\label{proof:reducibility}
Let $\Dc=\{\Dc_1,\dots,\Dc_s\}$ be the strong component decomposition of $\Dc$. This strong component decomposition is unique, because the strong connectivity of a digraph is an equivalence relation of the set of its vertices. Let $\Dc_s^*$ denote the strong component with maximal fractional dichromatic number. As $\Dc_s^*$ is an induced sub-digraph of $\Dc$, we have $\beta^{\rm SIC}(\Dc_s^*) \le \beta^{\rm SIC}(\Dc)$, because additional vertices do no reduce the broadcast rate.

Moreover, as $\Dc_s^*$ falls into the digraph classes in Theorem \ref{theorem:OA-optimal}, orthogonal access achieves the optimal DoF region, and in turn the optimal symmetric DoF for the TIM-MP problem. Thus, orthogonal access (fractional acyclic set coloring) also achieves the optimal broadcast rate of the corresponding SIC problem.
It follows that $\beta^{\rm SIC}(\Dc_s^*) =\chi_{A,f}(\Dc_s^*)$.

According to the definition of strong decomposition, the strong component with the maximum fractional dichromatic number dominates, i.e.,
\begin{align}
\chi_{A,f}(\Dc) &= \max_{ i=\{1,\dots,s\}} \chi_{A,f}(\Dc_i) \\
&=\chi_{A,f}(\Dc_s^*).
\end{align}

To sum up, we have
\begin{align}
\MoveEqLeft \chi_{A,f}(\Dc_s^*) = \beta^{\rm SIC}(\Dc_s^*) \nonumber\\
&\le \beta^{\rm SIC}(\Dc) \le \chi_{A,f}(\Dc_s) =\chi_{A,f}(\Dc_s^*).
\end{align}
where the second inequality is due to the achievability of orthogonal access. Thus, it follows that $\beta^{\rm SIC}(\Dc_s^*) = \beta^{\rm SIC}(\Dc)$, and those vertices $\Vc(\Dc) \backslash \Vc(\Dc_s^*)$ are reducible without reducing the broadcast rate.
This completes the proof.

\subsection{Proof of Theorem \ref{theorem:criticality}}
\label{proof:criticality}
Based on the clique and cycle inequalities, we have the lower bound of broadcast rate
\begin{align}
\beta^{\rm SIC} (\Dc)\ge \max \left \{ \abs{Q}, 1+\frac{1}{\abs{C}-1} \right\}
\end{align}
where the maximum is over all cliques $Q \in \Qc$ and dicycles $C \in \Cc$. So, the maximal clique or the minimal dicycle dominates.

Let us first consider the case when the dicycle-vertex incidence matrix of $\Dc$ is ideal. Let the induced dicycle $C_n$ be the shortest and unique one, where there does not exist any shorter or equal length induced dicycles. As $C_n$ is unique and the dicycle-vertex incidence matrix of $\Dc$ is ideal, the optimal broadcast rate is $\beta^{\rm SIC} (\Dc)=1+\frac{1}{n-1}$ which can be achieved by orthogonal access. The removal of the arc $e$ from $C_n$ either breaks the dicycle $C_n$ or forms a longer dicycle, both of which lead to the smallest induced dicycle being $C_m$ with $m \ge n+1$. The removal of $e$ will lead to a lower broadcast rate achieved by orthogonal access, because the resulting dicycle-vertex incidence matrix is still ideal and $\beta^{\rm SIC} (\Dc-e)\le \chi_{A}(\Dc-e)=1+\frac{1}{m-1} < \beta^{\rm SIC} (\Dc)$.

When $\Dc$ is a perfect digraph, the lower bound becomes $\beta^{\rm SIC} (\Dc)\ge \max_{Q \in \Qc}\{\abs{Q}\}$. Let the clique $Q_n$ be the maximal one and unique, where there does not exist any smaller or equal size cliques. So, $\beta^{\rm SIC} (\Dc)=n$ is achievable and also optimal. The removal of the arc $e$ from $Q_n$ does not break the perfectness and $\Dc-e$ is also a perfect digraph. The removal of $e$ will lead to a lower broadcast rate achieved by orthogonal access, because $\beta^{\rm SIC} (\Dc-e)= n-1 < \beta^{\rm SIC} (\Dc)$. This completes the proof.

\subsection{Proof of Theorem \ref{theorem:linear-opt}}
\label{proof:linear-opt}
According to the equivalence between TIM-MP and SIC with linear coding schemes, we focus in the following on TIM-MP, and apply to SIC accordingly.

For the transmitter $i$, we assign the message $W_i$ with a precoding matrix $\Vm_i$. By a bit abuse of notation, we also use $\Vm_i$ to represent the subspace spanned by the columns of $\Vm_i$. Thus, we have the linear symmetric DoF
\begin{align}
d_{\sym,l}  = \max \min_k \dim(\Vm_k), 
\end{align}
where $\dim(\cdot)$ is the normalized dimensionality, and the overall dimension satisfies $\dim(\cup_{k} \Vm_k) = 1$. Without loss of generality, we assume that $\dim(\Vm_k) = R$, $\forall~k$.

The achievability is due to fractional acyclic set coloring. It can be easily checked that $\chi_{A,f}(\Dc(\Cm_5^2)) = \frac{5}{2}$ and $\chi_{A,f}(\Dc(\Jm_3))=\frac{5}{2}$. 

Then, let us proceed to the converse proofs for linear coding schemes. For the instance $\Dc(\Jm_3)$, we have 
\begin{align}
\dim(\Vm_0 \cap \Vm_k)=0,
\end{align}
for all $k =1,2,3$ because nodes 0 and $k$ are fully conflicting. And, for any $i \neq j \neq k \in \{1,2,3\}$, we have
\begin{align}
\MoveEqLeft \dim (\Vm_i \cap \Vm_j) + \dim (\Vm_i \cap \Vm_k) \\
&= \dim (\Vm_i \cap (\Vm_j \cup \Vm_k)) + \dim (\Vm_i \cap \Vm_j \cap \Vm_k)\\
& = \dim (\Vm_i \cap (\Vm_j \cup \Vm_k)) \\
&\le \dim (\Vm_i) = R
\end{align}
because $\{1,2,3\}$ forms a dicycle and their subspaces should not have any overlap, i.e., $\dim (\Vm_i \cap \Vm_j \cap \Vm_k)=0$.
Thus, we have
\begin{align}
1 &= \dim(\cup_{k=0}^3 \Vm_k)\\
&= \sum_{\emptyset \neq S \subseteq \{0,1,2,3\}} (-1)^{\abs{S}-1} \dim (\cap_{k \in S} \Vm_k) \\
&=\sum_{k=0}^3 \dim(\Vm_k) - \binom{3}{2} \dim_{0<i<j \le 3}( \Vm_i \cap \Vm_j) \\
&\ge 4R - \frac{3}{2}R = \frac{5}{2} R
\end{align}
which yields $R \le \frac{2}{5}$. Together with the achievability, we have the optimal linear symmetric DoF $d_{\sym,l}(\Dc(\Jm_3))=\frac{2}{5}$.

Similarly, for the instance $\Dc(\Cm_5^2)$, we have
\begin{align}
\dim (\Vm_i \cap \Vm_{i+1}) = 0, \; \forall i
\end{align}
because the adjacent nodes are fully conflicting, and
\begin{align}
\dim(\Vm_i \cap \Vm_j \cap \Vm_k)=0, \forall i \ne j \ne k
\end{align}
because for any $i \ne j \ne k$, at least one of them is adjacent to another, and thus conflicts one another. For any $i \ne j$, we also have
\begin{align}
\MoveEqLeft \dim(\Vm_i \cap \Vm_j) + \dim(\Vm_i \cap \Vm_{j+1}) \\
&= \dim (\Vm_i \cap (\Vm_j \cup \Vm_{j+1})) + \dim (\Vm_i \cap \Vm_j \cap \Vm_{j+1})\\
&\le \dim(\Vm_i) = R.
\end{align}
Thus, we have
\begin{align}
1 &= \dim(\cup_{k=1}^5 \Vm_k)\\
&= \sum_{\emptyset \neq S \subseteq \{1,2,\dots 5\}} (-1)^{\abs{S}-1} \dim (\cap_{k \in S} \Vm_k) \\
&=\sum_{k=1}^5 \dim(\Vm_k) - \sum_{i=1}^5 \dim( \Vm_i \cap \Vm_{i+2}) \\
&\ge 5R - \frac{5}{2}R = \frac{5}{2} R
\end{align}
which yields $R \le \frac{2}{5}$. Together with the achievability, we have $d_{\sym,l}(\Dc(\Cm_5^2))=\frac{2}{5}$.

\subsection{Proof of Theorem \ref{theorem:4-user}}
\label{proof:4-user}
Given the equivalence between TIM-MP and SIC with linear coding schemes, we refer to both of them interchangeably.

According to vertex-reducibility, we only have to focus on the topologies whose conflict digraphs are strongly connected (and thus irreducible), because otherwise the 4-user instances can be reduced to 3-user ones, which have been already proven that orthogonal access achieves the optimal DoF/capacity region. According to arc-criticality, we only have to consider the arcs that belong to at least one induced dicycle, because otherwise the arcs are not critical and can be removed without changing the optimal DoF/capacity region.

For these irreducible topologies, we only have to focus on the imperfect ones, because orthogonal access achieves the optimal DoF/capacity region for perfect digraphs. According to the definition of perfect digraphs, we only have to consider the conflict graphs with dicycles $C_3$ or $C_4$ as induced sub-digraph. For the case that contains $C_4$ as induced sub-digraph, we have only one topology which is exactly $C_4$ (Fig. \ref{fig:4-user}(a)), and it was proven that orthogonal access achieves the optimal DoF/capacity region. For the case that contains $C_3$ as induced sub-digraph, we can restrict ourselves to a few cases. Assume without loss of generality that vertices 1,2, and 3 form $C_3$. Then, in view of the fact that vertex 4 is irreducible and the arcs involving it are critical, there are the following possibilities for the connection between vertex 4 and vertices in $C_3$: (1) vertex 4 forms another length-3 dicycle with any two vertices of $C_3$, as in Fig. \ref{fig:4-user}(b); (2) vertex 4 forms another length-3 dicycle with any two vertices of $C_3$ and a length-2 dicycle with the third one in $C_3$, as in Fig. \ref{fig:4-user}(c); (3) vertex 4 forms 1, 2, or 3 length-2 dicycles with some vertices in $C_3$, respectively, as in Fig. \ref{fig:4-user}(d-f). For the digraphs in Fig. \ref{fig:4-user}(a,b,d,e), it can be checked that the dicycle-vertex incidence matrices are ideal, and thus orthogonal access achieves the optimal DoF/capacity region, and in turn the symmetric DoF/capacity. For Fig. \ref{fig:4-user}(c), the optimal symmetric DoF is $d_{\sym}=\frac{1}{2}$ that can be achieved by orthogonal access, although the dicycle-vertex incidence matrix is not ideal. For Fig. \ref{fig:4-user}(f), the dicycle-vertex incidence matrix is $\Jm_3$ and thus non-ideal. From Theorem \ref{theorem:linear-opt}, we have proved that orthogonal access achieves the linear optimal symmetric DoF of the TIM-MP problem, and also the linear optimal broadcast rate of the corresponding SIC problem.

To sum up, we conclude that, for the TIM-MP/SIC problems up to 4 users, orthogonal access achieves the linear optimal symmetric DoF/rate. This completes the proof.
\begin{figure}[htb]
 \centering
\includegraphics[width=0.8\columnwidth]{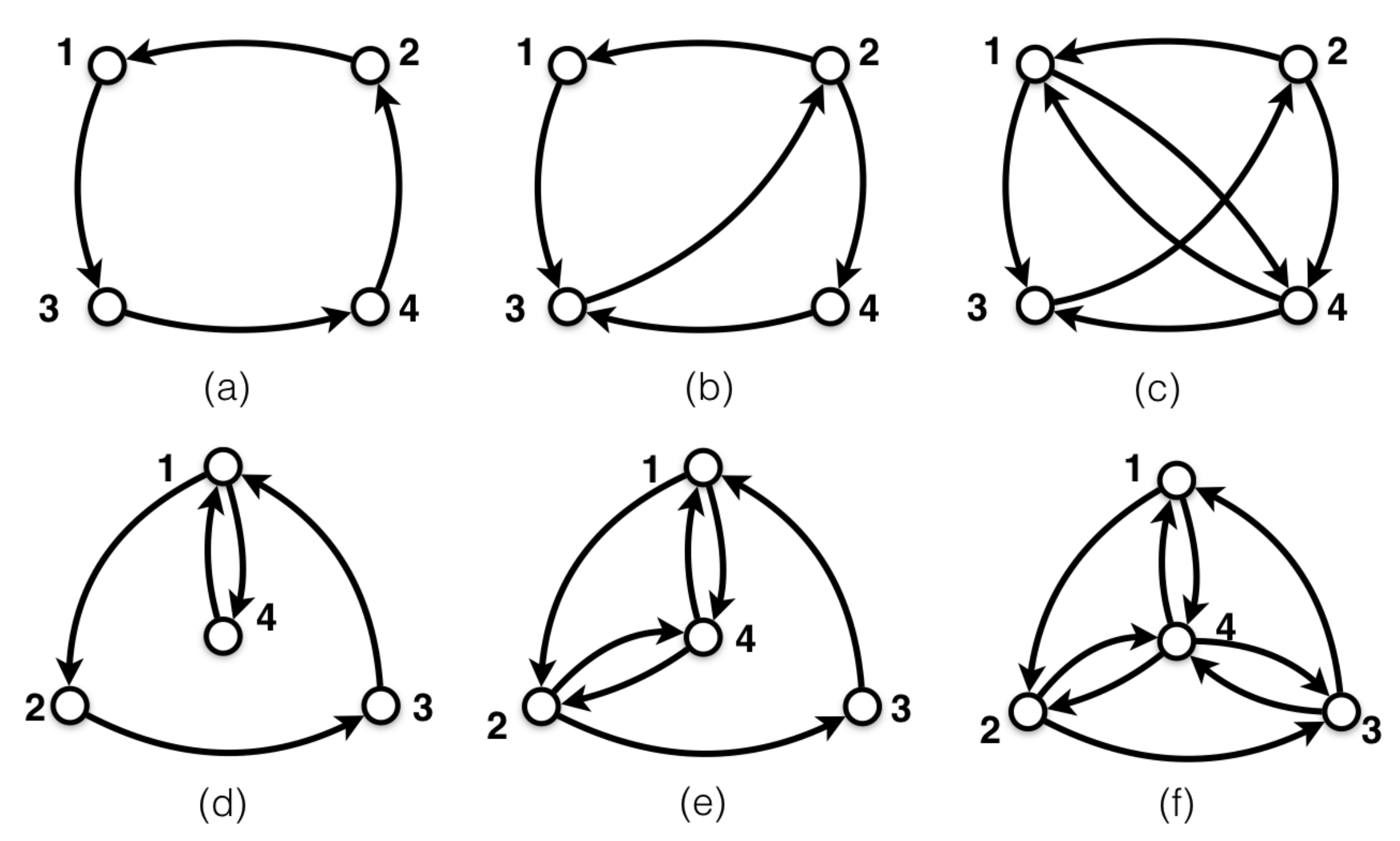}
\caption{The conflict digraphs $\Dc$ of 4-user network that are not perfect.}
\label{fig:4-user}
\end{figure}

\subsection{Proof of Theorem \ref{theorem:helpful}}
\label{proof:helpful}
Let $(i,j)$ be the arc corresponds to the message passing $i \to j$, and $\bar{\Dc}[\Sc] \cup (i,j)$ be the induced new dicycle where $i,j \in \Sc$. It immediately follows that the sub-digraph $\bar{\Dc}[\Sc]$ is acyclic. Due to MAIS outer bound (see Appendix \ref{appendix:index}), we have the achievable DoF tuple before adding $(i,j)$ should satisfy
\begin{align}
\sum_{k \in \Sc} d_{k} \le 1.
\end{align}
After adding the arc $(i,j)$, the sub-digraph induced by $\Sc$ becomes a dicycle. As such, $d_k=\frac{1}{\abs{\Sc}-1}$, $\forall~k \in \Sc$ is achievable, leading to a larger sum DoF. As message passing does not reduce the achievable DoF of other messages not in $\Sc$. The addition of arc $(i,j)$ increases the DoF region, and hence is helpful. 

\subsection{Proof of Corollary \ref{cor:chordal-helpful}}
\label{proof:chordal-helpful}
As a special case, the sufficiency follows exactly as Theorem \ref{theorem:helpful}.
Then, we focus on the necessity.
By contraposition, we show that if the addition of the corresponding arc in $\bar{\Dc}$ does not form any new dicycles, then such an arc addition will not change the DoF region. According to \cite[Theorem 1]{Yi:Fractional}, the optimal DoF region is fully characterized by the clique inequalities of the underlying undirected conflict graph $U(\Dc)$. The addition of an arc in $\bar{\Dc}$ is equivalent to the removal of the corresponding arc in $\Dc$. As the arc addition does not form new dicycles, it follows that (1) the removed arc in $\Dc$ should not be uni-directed, because otherwise it will form new dicyles in $\bar{\Dc}$, and (2) the removed arc in $\Dc$ should not introduce new dicycles in $\Dc$ as induced sub-digraphs, because otherwise it will also results in a new dicycle in $\bar{\Dc}$ as well. So, we conclude that the arc removal in $\Dc$ does not change the underlying undirected conflict graph $U(\Dc)$ and the resulting network remains a chordal bipartite network, and thus the DoF region is not changed. By contraposition, this completes the proof of necessity.



\end{document}